\documentclass[12pt]{article}

\usepackage{hyperref}
\usepackage{amsmath}
\usepackage{amssymb}
\usepackage{amsthm}
\usepackage{amsfonts}
\usepackage{latexsym}
\usepackage{cite}
\usepackage{psfrag}
\usepackage{epsfig}
\usepackage{graphicx}
\usepackage{mathrsfs}
\usepackage{url}
\usepackage{color}
\usepackage{cancel}
 \usepackage{eufrak}
 \usepackage[pdftex]{lscape}
 \usepackage{lscape}

\usepackage{multirow}
\usepackage{multicol}
\usepackage{hyperref}

\usepackage{scalerel}
\usepackage{tikz}
\usetikzlibrary{svg.path}

\definecolor{orcidlogocol}{HTML}{A6CE39}
\tikzset{
  orcidlogo/.pic={
    \fill[orcidlogocol] svg{M256,128c0,70.7-57.3,128-128,128C57.3,256,0,198.7,0,128C0,57.3,57.3,0,128,0C198.7,0,256,57.3,256,128z};
    \fill[white] svg{M86.3,186.2H70.9V79.1h15.4v48.4V186.2z}
                 svg{M108.9,79.1h41.6c39.6,0,57,28.3,57,53.6c0,27.5-21.5,53.6-56.8,53.6h-41.8V79.1z M124.3,172.4h24.5c34.9,0,42.9-26.5,42.9-39.7c0-21.5-13.7-39.7-43.7-39.7h-23.7V172.4z}
                 svg{M88.7,56.8c0,5.5-4.5,10.1-10.1,10.1c-5.6,0-10.1-4.6-10.1-10.1c0-5.6,4.5-10.1,10.1-10.1C84.2,46.7,88.7,51.3,88.7,56.8z};
  }
}

\newcommand\orcidicon[1]{\href{https://orcid.org/#1}{\mbox{\scalerel*{
\begin{tikzpicture}[yscale=-1,transform shape]
\pic{orcidlogo};
\end{tikzpicture}
}{|}}}}

\usepackage{hyperref} 

\usepackage{longtable}

\makeatletter
\@addtoreset{equation}{section}
\makeatother

\makeatletter
\@addtoreset{table}{section}
\makeatother

\newtheorem{theorem}{Theorem}[section]
\newtheorem{lemma}[theorem]{Lemma}
\newtheorem{corollary}[theorem]{Corollary}

\newtheorem{conjecture}[theorem]{Conjecture}
\theoremstyle{definition}
\newtheorem{definition}[theorem]{Definition}
\newtheorem{notation}[theorem]{Notation}

\newtheorem{example}[theorem]{Example}
\newtheorem{problem}[theorem]{Open Problem}

\newcommand\C{\mathcal{C}}
\newcommand\V{\mathcal{V}}

\newcommand\Oc{\mathcal{O}}

\newcommand\cf{\mathbf{c}}

\newcommand\e{\mathbf{e}}
\newcommand\vf{\mathbf{v}}
\newcommand\x{\mathbf{x}}

\newcommand\F{\mathbb{F}}
\newcommand\Nb{\mathbb{N}}

\newcommand\Pb{\mathbb{P}}
\newcommand\Z{\mathbb{Z}}

\newcommand\Dk{\mathfrak{D}}

\newcommand\Rk{\mathfrak{R}}

\newcommand\PG{\mathrm{PG}}

\newcommand\B{\mathscr{B}}
\newcommand\Ns{\mathscr{N}}
\newcommand\Os{\mathscr{O}}
\newcommand\Ps{\mathscr{P}}

\newcommand\db{\displaybreak[3]}

\begin{document}
\title{Weight distributions of cosets of weight 2 of the generalized doubly extended Reed--Solomon codes
\date{}
}
\maketitle

\begin{center}
{\sc Alexander A. Davydov \orcidicon{0000-0002-5827-4560}}\\
 {\sc\small Kharkevich Institute for Information Transmission Problems}\\
 {\sc\small Russian Academy of Sciences,
Moscow, 127051, Russian Federation}\\
 \emph{E-mail address:} alexander.davydov121@gmail.com\medskip\\
 {\sc Stefano Marcugini \orcidicon{0000-0002-7961-0260} and
 Fernanda Pambianco \orcidicon{0000-0001-5476-5365}}\\
 {\sc\small Department of  Mathematics  and Computer Science, University of Perugia,}\\
 {\sc\small Perugia, 06123, Italy}\\
 \emph{E-mail address:} \{stefano.marcugini, fernanda.pambianco\}@unipg.it
\end{center}

\textbf{Abstract.} We consider the weight distributions of the cosets of weight 2 of the generalized $[q+1,q+2-d,d]_q$ doubly extended Reed--Solomon codes (GDRS) of length $q+1$, dimension $q+2-d$, minimum distance $d\ge5$, over the finite field $\mathbb{F}_q$ with $q$ elements. GDRS codes form a broad class of MDS codes. For a GDRS code, we say that Case S occurs if the weight distribution for all  cosets of weight 2 is the same or otherwise, Case NS occurs. For Case S, the weight distribution is known; however, any sufficient condition for the occurrence of Case S remained an open problem. We prove that if $q-1$ and $d-2$ are coprime then Case S holds, i.e. the problem is solved. Furthermore, we note that in Case S, the GDRS code is 2-regular.

Also, we introduce two new,  as far as is known to the authors, open equivalent combinatorial problems for finite fields $\mathbb{F}_q$ (Problem $A_{q,\mu}^\times$) and for rings $\mathbb{Z}_\mathfrak{R}$ of integers modulo $\mathfrak{R}$ (Problem $A_{\mathfrak{R},\mu}^+$), where $\mu$ is a parameter. In particular, Problem $A_{\mathfrak{R},\mu}^+$ is as follows: for each element $\lambda$ of $\mathbb{Z}_\mathfrak{R}$, determine the number of all possible $\mu$-tuples $\{\lambda_1,\lambda_2,\ldots,\lambda_{\mu}\}$, each of which consists of $\mu$ distinct  elements $\lambda_j$ of $\mathbb{Z}_{\Rk}$ such that their sum in $\mathbb{Z}_{\Rk}$ is equal to $\lambda$.  Open Problems $A_{q,\mu}^\times$ and $A_{\mathfrak{R},\mu}^+$ are interesting in their own right and, moreover, we proved that their solutions allow us to obtain the weight distributions for Case NS, taking $\mu=d-2$ and $\mathfrak{R}=q-1$.
To solve Problem $A_{\mathfrak{R},\mu}^+$, we found a universal method, connected with the values of $\mathfrak{R}$ and $\mu$, using orbits of elements in $\mathbb{Z}_{\Rk}$ and then we solved the problem for many pairs $\mathfrak{R},\mu$, obtaining the needed weight distributions for the corresponding pairs $q=\mathfrak{R}+1,d=\mu+2$.



\section{Introduction}\label{sec_Intro}
Let $\F_{q}$ be the Galois field with $q$ elements, $\F_{q}^*=\F_{q}\setminus\{0\}$. Let $\F_{q}^{n}$ be
the space of $n$-dimensional vectors over ${\mathbb{F}}_{q}$.  We denote by  $[n,k,d]_{q}R$ an $\F_q$-linear code of length $n$, dimension $k$, minimum distance $d$, and covering radius $R$. In this notation, one may omit $R$ and call the ``minimum distance''  simply the ``distance''. If $d=n-k+1$, the code is maximum distance separable (MDS). The $[q+1,k,q+2-k]_q$ generalized doubly extended Reed--Solomon codes (GDRS) form an important class of MDS codes. For many values of $q$ and $k$, any $[q+1,k,q+2-k]_q$ MDS code is a GDRS code; see, e.g., \cite[Section~11]{Roth}. Thus, cosets of the GDRS code, considered in this paper,
are, in fact, cosets of a broad class of MDS codes.
  For an introduction to coding theory, see \cite{Bier,Blahut,HufPless,MWS,Roth}. For preliminaries, see Section \ref{sec:prelimin}.

Let $\PG(N,q)$ be the $N$-dimensional projective space over~$\F_q$. An $n$-\emph{arc} in  $\PG(N,q)$ is a
set of $n$ points such that no $N +1$ points lie in the same hyperplane of $\PG(N,q)$. An $n$-arc is complete if it is not contained in an $(n+1)$-arc. Arcs and MDS codes are equivalent objects  \cite{Ball-book2015,BallLav,EtzStorm2016,LandSt,MWS}.
 For an introduction to projective spaces over finite fields and to the connections between projective geometry, coding theory, and combinatorics, see \cite{Ball-book2015,EtzStorm2016,Hirs_PGFF,HirsStor-2001,HirsThas-2015,LandSt}.

A \emph{coset} of a code is a translation of the code by some vector, say $\vf$. A coset $\V$ of an $[n,k,d]_{q}$ code $\C$ can be represented as
\begin{equation}\label{eq1:coset}
  \V=\vf+\C=\{\x\in\F_q^n\,|\,\x=\cf+\vf,\cf\in \C\}\subset\F_q^n
\end{equation}
 where $\vf\in \V$ is a vector  fixed for the given representation and $\cf$ is a codeword; see \cite{Blahut,HufPless,MWS,Roth}. The \emph{weight $W$ of a coset} is the smallest Hamming weight of any vector in the coset. For preliminaries, see Section \ref{subsec:cosets}.

  The weight distributions of code cosets, their classification, the number of the cosets with distinct distributions, are interesting in their own right; they are important combinatorial properties of a code. In particular, the coset weight distribution serves to estimate decoding performance, e.g., for the \emph{bounded distance decoder} \cite[Section 4.3]{Blahut2008},  \cite{CheungIEEE1989,CheungIEEE1992}, \cite[Section 2.1, p.\ 1120]{DMP_MDScosets}, see Section \ref{subsec:cosets} for details. Knowledge of the weight distributions
of code cosets gives information on the distance distribution of the code itself. There are
many papers connected with various aspects of the weight distributions of code cosets; see, e.g.,
\cite{Blahut}, \cite[Section 7]{HufPless}, \cite[Sections 5.5, 6.6, 6.9]{MWS}, \cite[Section 10]{HandbookCodes},
\cite[Section~6.3]{DelsarteBook}, \cite[Section 1.4.6]{Klove},
\cite{Blahut2008,CheungIEEE1989,CheungIEEE1992,AsmMat,BlokPelSzo,Bonneau1990,CharpHelZin,DMP_IntegrWeight1_2Cosets,%
DMP_CosetsRScod4,DMP_MDScosets,Delsarte4Fundam,DelsarteLeven,Helleseth,JurrPellik,KasLin,KaipaIEEE2017,%
MW1963,Schatz1980,XuXu2019,ZDK_DHRIEEE2019}, and the references therein.

In \cite{CheungIEEE1989}, for $[n,k,d]_{q}$ MDS codes, the integrated weight distribution for the union of all cosets of weight $W\le\lfloor(d-1)/2\rfloor$ is obtained. In \cite{DMP_IntegrWeight1_2Cosets}, using results from \cite{CheungIEEE1989}, the weight distributions for the unions of all cosets of weight 1 and all cosets of weight 2 of MDS codes are given. In \cite{DMP_CosetsRScod4}, the weight distributions of all cosets  (without taking any unions) of the specific $[q+1,q-3,5]_q3$ GDRS code are obtained. In \cite{BlokPelSzo}, the coset leader weight enumerator for a $[q+1,q-3,5]_{q^m}4$ code, $m\ge2$, is considered.

In a few works, see, e.g., \cite{Bonneau1990,Delsarte4Fundam,HufPless,MWS}, it is shown that for an MDS code of distance~$d$, the weight distribution of any coset is uniquely determined if, in the coset, the numbers of vectors of weights $1,2,\ldots,d-2$ are known. Methods to obtain the weight distribution of a coset using the $d-2$ known numbers are considered in \cite[Section 7]{HufPless}, \cite[Section 10]{HandbookCodes}, \cite{Bonneau1990,DMP_CosetsRScod4,DMP_MDScosets}. The approach of \cite{HufPless,HandbookCodes} can be used for various codes, not necessarily MDS.

In Bonneau's paper \cite{Bonneau1990}, for MDS codes, a formula is proposed which is a remarkable direct relation between the $d-2$ known numbers of vectors in a coset and the full weight distribution of this coset. The approach of \cite{Bonneau1990} is simpler than that of \cite{HufPless,HandbookCodes}.

In \cite{DMP_CosetsRScod4}, for a $[q+1,q-3,5]_q3$ GDRS code, the numbers of vectors of weight $d-2=3$ in its cosets are obtained by a projective-geometric method based on the results of \cite{BDMP_TwistCubFFA} and then the results of \cite{Bonneau1990} are used to obtain the full weight distributions of the cosets.

In \cite{DMP_MDScosets}, the Bonneau formula of \cite{Bonneau1990} is transformed into a simpler form; also for many cases, further simplifications are made. The improved tools are applied to
obtain the weight distributions of the cosets of distinct MDS codes. In particular, for $[q+1,q-2,4]_qR$ GDRS codes with $R=2$, $q$ odd, and $R=3$, $q$ even, the weight distributions of the cosets are obtained in full, see \cite[Section 6]{DMP_MDScosets}.

For $[q+1,q+2-d,d]_qR$ GDRS codes with $d\ge5$, many interesting and useful results are obtained in
\cite[Sections 4, 5]{DMP_MDScosets}; see Section \ref{subsec2:WD} for details. However, for the weight distributions of the cosets of weight 2 many problems remained open. In this paper, we solve these problems in part.

Taking into account that a GDRS code is MDS, for a $[q+1,q+2-d,d]_qR$ GDRS code with $d\ge5$, the weight distribution of its coset $\V^{(2)}$ of weight 2 is completely determined if the number $B_{d-2}(\V^{(2)})$ of vectors of weight $d-2$ in $\V^{(2)}$ is known, see \cite[Theorem 5.1]{DMP_MDScosets} and Theorem \ref{th2:W=2}(i) with \eqref{eq2:Bon_w=2}.

For a GDRS code, we say that Case S occurs if the weight distribution for all cosets of weight 2 is the same or otherwise, Case NS occurs.

For Case S, according to the literature, for all the cosets $\V^{(2)}$, we have $B_{d-2}(\V^{(2)})=\frac{1}{q-1}\binom{q-1}{d-2}$, see \cite[Theorem 5.3]{DMP_MDScosets} and Theorem \ref{th2:W=2}(iii) with \eqref{eq2:wd2_ident}; i.e. the required weight distributions are known. However, any sufficient condition, for the occurrence of Case S, remained an open problem. We proved that if $q-1$ and $d-2$ are coprime then Case S holds, i.e. the problem is solved.

Furthermore, in Theorem \ref{th5:WD w2coset}(iii), we note that in Case S, the GDRS code is \emph{$2$-regular}; see \cite{BorgRiZin2019PIT,BorgRiZin2025arX,Delsarte4Fundam,GoethTilb1975UP} and Section \ref{subsec24:reguar} with Definition \ref{def2:tregul} for definitions.

For Case NS, we prove that if a coset $\V^{(2)}$ of weight 2 of a $[q+1,q+2-d,d]_qR$ GDRS code  with $d\ge5$ has a coset leader vector with elements $\gamma_1,\gamma_2\in\F_q^*$ in positions $j_1,j_2$ then the value of $B_{d-2}(\V^{(2)})$ does not depend on the positions $j_1,j_2$ and is equal to the number of all possible $(d-2)$-tuples $\{\alpha_1,\alpha_2,\ldots,\alpha_{d-2}\}$, each of which consists of $d-2$ distinct elements $\alpha_j$ of $\F_q^*$ such that their product is equal to the element $-\gamma_2/\gamma_1$ of $\F_q^*$, see Section \ref{sec:wd-2vectors} with Theorem \ref{th3:Bd-2}. As far as is known to the authors, the corresponding general combinatorial problem is not considered in the literature. We introduce it as Open Problem $A_{q,\mu}^\times$ for Galois fields $\F_q$; see Open Problem \ref{openprobAx} and Definition \ref{def3:Gamma} with \eqref{eq3:hatPb gamma}, \eqref{eq3:Pb gamma}.

 To find solutions of Open Problem $A_{q,\mu}^\times$, the following, equivalent, and in fact somewhat broader, open problem turned out to be more convenient. Open Problem $A_{\Rk,\mu}^+$ for rings $\Z_\Rk$ of integers modulo $\Rk$ reduces to the following: for each element $\lambda$ of $\Z_\Rk$, determine the number of all possible $\mu$-tuples $\{\lambda_1,\lambda_2,\ldots,\lambda_{\mu}\}$, each of which consists of $\mu$ distinct  elements $\lambda_j$ of $\Z_{\Rk}$ such that their sum in $\Z_{\Rk}$ is equal to $\lambda$;
see Open Problem \ref{openprobA+}  and Definition \ref{def3:Lambda} with \eqref{eq3:hatPk lambda}, \eqref{eq3:Pk lambda}. As far as is known to the authors, Open Problem $A_{\Rk,\mu}^+$ is not considered in the literature.

 Open Problem $A_{q,\mu}^\times$ and Open Problem $A_{\Rk,\mu}^+$ are interesting in their own right and, moreover, their solutions allow us to obtain the weight distributions for Case NS, taking $\mu=d-2$ and $\Rk=q-1$.
 To solve Problem $A_{\Rk,\mu}^+$, we found a universal method, connected with the values of $\Rk$ and $\mu$, using orbits of elements in $\Z_\Rk$ and then we solved the problem for many pairs $\Rk,\mu$, obtaining the needed weight distributions for the corresponding pairs $q=\Rk+1,d=\mu+2$.

The paper is organized as follows. Section \ref{sec:prelimin} contains preliminaries; therefore, the paper is mostly self-contained. In Section \ref{sec:wd-2vectors}, the number $B_{d-2}(\V^{(2)})$ of vectors of weight $d-2$ in the cosets of weight 2 of GDRS codes with distance $d\ge5$ is considered, Open Problems $A_{q,\mu}^\times$ and $A_{\mathfrak{R},\mu}^+$ are formulated and their sufficiency for obtaining $B_{d-2}(\V^{(2)})$ is proved. In Section~\ref{sec4:orbits}, for many pairs $\mathfrak{R},\mu$, solution of Open Problem $A_{\mathfrak{R},\mu}^+$ is obtained; in other words, a partial solution of Problem $A_{\Rk,\mu}^+$ is obtained. In Section \ref{sec5:WD2}, using the results above, we obtained the weight distribution of cosets of weight 2 of the $[q+1,q+2-d,d]_q$ GDRS codes for many pairs $q=\Rk+1,d=\mu+2$.

\section{Preliminaries}\label{sec:prelimin}
We introduce notation and recall known definitions and properties of linear codes, their cosets,  and the normal rational curves in the projective spaces $\PG(N,q)$, see  \cite{Bier,Ball-book2015,BallLav,BDMP_AlmComplArc,BartGiulPlat,Blahut,Blahut2008,Bonneau1990,CheungIEEE1989,CheungIEEE1992,%
DMP_IntegrWeight1_2Cosets,DMP_MDScosets,Delsarte4Fundam,GabKl,HufPless,HandbookCodes,MWS,Roth} and the references therein.


\subsection{Cosets of a linear code; the bounded distance decoder}\label{subsec:cosets}
\begin{notation}\label{not1}
 For an $[n,k,d]_{q}$ code $\C$ and its cosets $\V$ of the form \eqref{eq1:coset}, we use the following notations and definitions:
\begin{align*}
&t(\C)=\left\lfloor(d-1)/2\right\rfloor&&\text{the number of errors correctable by the code }\C;\db\\
&\text{weight of a vector}&&\text{Hamming weight of the vector;}\db\\
&wt(\x)&&\text{the Hamming weight of a vector  $\x\in\F_q^n$;}\db\\
&\#M&&\text{the cardinality of a set }M;\db\\
&\triangleq&&\text{the sign “equality by definition”};\db\\
&A_w(\C)&&\text{the number of codewords of weight $w$ of the code }\C;\db\\
&S(\C)&&\text{the set of non-zero weights in}\,\C;~S(\C)=\{w>0|A_w(\C)\ne0\};\db\\
&s(\C)=\#S(\C)&&\text{the number of non-zero weights in }\C;\db\\
&\C^\bot&&\text{the $[n,n-k,d^\bot]_q$ code dual to the }[n,k,d]_{q}\text{ code }\C;\db\\
&\vf+\C&&\text{the coset of $\C$ of the form }\eqref{eq1:coset};\db\\
&\cf_w&&\text{a codeword of weight $w$ of $\C$};~ wt(\cf_w)=w,~\cf_w\in\C;\db\\
&B_w(\V)&&\text{the number of vectors of weight $w$ in the coset }\V;\db\\
&\text{a coset leader}&&\text{a vector in the coset having the smallest Hamming weight;}\db\\
&\text{the weight of a coset}&&\text{the smallest Hamming weight of any vector in the coset;}\db\\
&\V^{(W)}&&\text{a coset of weight }W;~~B_w(\V^{(W)})=0\text{ if }w<W;\db\\
&B_W(\V^{(W)})&&\text{the number of all coset leaders in a coset }\V^{(W)};\db\\
&\Nb^{(W)}_\Sigma(\C)&&\text{the total number of the cosets of weight $W$ of the code }\C;\db\\
&\B_w^{\Sigma}(\V^{(W)})&&\text{the overall number of vectors of weight $w$ in all cosets}\db\\
&&&\text{of weight }W;\db\\
&\mathbf{H}(\C)&&\text{an $(n-k)\times n$ parity check matrix of }\C;\db\\
&tr&&\text{sign of the transposition};\db\\
&\mathbf{H}(\C)\x^{tr}&&\text{the \emph{syndrome} of a vector }\x\in\F_q^n,~~\mathbf{H}(\C)\x^{tr}\in \F_q^{n-k}.
\end{align*}
\end{notation}
 All vectors of a coset have the same syndrome; it is called the \emph{coset syndrome}. There is a one-to-one correspondence between cosets and syndromes. The code $\C$  is the coset of weight zero. The syndrome of $\C$ is the zero vector of $ \F_q^{n-k}$.

 If $W\leq t(\C)$ we have $B_W(\V^{(W)})=1$ and $ \Nb^{(W)}_\Sigma(\C)=\binom{n}{W}(q-1)^W$.
  If $W> t(\C)$, then $B_W(\V^{(W)})\ge1$, i.e. a  vector of minimal weight is not necessarily unique.

The covering radius of an
$[n,k,d]_{q}$ code $\C$ is the least integer $R$ such that the space $\F_{q}^{n}$ is covered by
Hamming spheres of radius $R$ centered at the codewords. Every column of $\F_{q}^{n-k}$ is equal to a linear combination of at most $R$ columns of a parity check matrix of $\C$.
 The covering radius $R$ of the code $\C$ is equal to the maximum weight of a coset of $\C$.
 For an $[n,k,d]_qR$ MDS code we have $R\le d-1$ \cite{BartGiulPlat,GabKl}.

For a fixed~$W$, we call the set $\{\B_w^{\Sigma}(\V^{(W)})|w=0,1,\ldots,n\}$
\emph{integral weight spectrum} of the code cosets of weight $W$, see \cite{CheungIEEE1989,CheungIEEE1992,DMP_IntegrWeight1_2Cosets} and the references therein.

The coset weight distribution, including integral weight spectrum, is important for estimating the \emph{bounded distance decoder} which acts as follows \cite[Section 4.3]{Blahut2008}, \cite{CheungIEEE1989,CheungIEEE1992}, \cite[Section 2.1, p.\ 1120]{DMP_MDScosets}.
For an $[n,k,d]_{q}$ code $\C$, let $\cf\in \C$ be a sent word, $\e$ be an
error vector, and $\x=\cf+\e$ be the received word. Let $\tau$ be some given distance. If $\x$ belongs to a coset of weight $>\tau$, the decoder declares the detection of an uncorrectable error. If $\x$ belongs  to a coset of weight $\le \tau$, the decoder puts the coset leader as an error vector $\e$, if the leader is unique, otherwise it forms a list of all the leaders as possible errors. This leads to incorrect decoding if, in fact, $wt(\e)>\tau$. Thus, all vectors of weight $>\tau$ belonging to the cosets of weight $\le \tau$ (resp. $> \tau$) give rise to incorrect decoding (resp. to the detection of an uncorrectable error). If $\tau=t(\C)=\lfloor(d-1)/2\rfloor$, the decoder is called the \emph{decoder up to half of minimum distance}.
For the corresponding probabilities, see  \cite{CheungIEEE1989,CheungIEEE1992}.

\subsection{MDS, GDRS,
and Hamming codes; normal rational curves}\label{subsec:MDS_Prelim}


The weight distribution  of an $[n,k,d=n-k+1]_q$ MDS code $\C$ has the following form; see, e.g.,
\cite[Theorem 7.4.1]{HufPless}, \cite[Theorem 11.3.6]{MWS}:
\begin{equation}\label{eq2:Aw(C)}
 A_{w}(\C)=0\text{ if }w<d;~A_{w}(\C)=\binom{n}{w}\sum_{j=0}^{w-d}(-1)^j\binom{w}{j}(q^{w-d+1-j}-1)\text{ if }w\ge d.
\end{equation}

A $(d-1)\times (q+1)$ parity check matrix $\mathbf{H}_d$ of the $[q+1,q+2-d,d]_q$ GDRS code with $d\ge3$ can be represented \cite[Section 5.1]{Roth} as
\begin{equation}\label{eq2:HDRS}
\mathbf{H}_d=\left[ \begin{array}{cccccc}
        1        &1        &\ldots&1&1        &0      \\
        m_1      &m_2      &\ldots&m_{q-1}&m_q      &0     \smallskip \\
        m_1^2    &m_2^2    &\ldots&m_{q-1}^2&m_q^2    &0     \\
        \ldots   &\ldots   &\ldots&\ldots   &\ldots&\ldots\\
        m_1^{d-3}&m_2^{d-3}&\ldots&m_{q-1}^{d-3}&m_q^{d-3}&0     \smallskip  \\
        m_1^{d-2}&m_2^{d-2}&\ldots&m_{q-1}^{d-2}&m_q^{d-2}&1       \\
       \end{array}\right]\times
       \left[\begin{array}{ccccc}
         v_1 & 0 & \ldots & 0 & 0 \\
         0 & v_2 & \ldots & 0 & 0 \\
        \ldots & \ldots & \ldots & \ldots & \ldots \\
         0 & 0 & \ldots& v_{q} & 0 \\
         0 & 0 & \ldots & 0 & v_{q+1}
                \end{array}\right],
\end{equation}
where for $1\le i\le q-1$ we have $m_i\in\F_{q}^*$, $m_i\ne m_j$ if $i\ne j$, i.e. $\{m_1,\ldots,m_{q-1}\}=\F_{q}^*$; also $m_q=0$, i.e. the $q$-th column of the left matrix is
$[1,0,0,\ldots,0]^{tr}$ and $\{m_1,\ldots,m_q\}=\F_{q}$. Finally, $v_i\in\F_q^*$, the elements $v_i$ do not have to
be distinct. If $v_1=\ldots=v_{q+1}=1$, the matrix $\mathbf{H}_d$ reduces to the left matrix and defines
the \emph{normalized} GDRS code \cite{Roth}.


 The $[\frac{q^m-1}{q-1},\frac{q^m-1}{q-1}-m,3]_q1$ Hamming code is well known \cite{HufPless,MWS,Roth}. For $m=2$, it is the $[q+1,q-1,3]_q1$ GDRS code with the parity check matrix $\mathbf{H}_3$ of \eqref{eq2:HDRS}.

In $\PG(N,q)$, $2\le N\le q-2$, a \emph{normal rational curve} is a $(q+1)$-arc projectively equivalent to the arc
$\{(1,t,t^2,\ldots, t^N):t\in \F_q\}\cup \{(0,\ldots,0 ,1)\}$ where $(1,t,\ldots,t^N)$ and $(0,\ldots,0 ,1)$ are points in homogeneous coordinates.  The points  of a normal rational curve in $\PG(N,q)$,
treated as columns, define a parity check matrix of a $[q + 1,q-N,N + 2]_qR$ GDRS code \cite{EtzStorm2016,LandSt,Roth}, cf. \eqref{eq2:HDRS}.

If in $\PG(N,q)$, the normal rational curve is a \emph{complete $(q+1)$-arc}, then the corresponding $[q + 1,q-N,N + 2]_qR$ GDRS code $\C$ cannot be extended to a $[q + 2,q-N+1,N + 2]_q$ MDS code, i.e. $\C$ has covering radius $R=d-2=N$.

The following conjectures are well known.

\begin{conjecture} \label{conj2:NRC}Let $2\le N\le q-2$. In $\PG(N,q)$, every normal rational curve is a complete $(q+1)$-arc except for the cases when  $q$ is even and $N\in\{2,q-2\}$, in which one point can be added to the curve.
\end{conjecture}
\begin{conjecture} \label{conj2:MDS} \emph{(MDS conjecture)} Let $2\le N\le q-2$. An  $[n,n-N-1,N+2]_q$ MDS code (or, equivalently, an $n$-arc in $\PG(N,q)$) has length $n\le q+1$ except for the cases when $q$ is even and $N\in\{2,q-2\}$, in which $n\le q+2$.
\end{conjecture}
 If the MDS conjecture holds for some pair $(N,q)$ then Conjecture \ref{conj2:NRC} holds too, but in general, the reverse is not true.
For the pairs $(N,q)$ for which MDS conjecture is proved, see
 \cite{Ball-book2015,BallLav} and the references therein.
 For the pairs $(N,q)$ for which Conjecture \ref{conj2:NRC} is proved, see \cite{BDMP_AlmComplArc}  and the references therein including \cite{St_1992_ComplNRC}.

 The couples mentioned are relatively numerous. Hence, the $[q+1, q+2-d,d]_qR$ GDRS codes with $R=d-2$ are an important class of MDS codes of covering radius $R=d-2$.

 \subsection{The weight distribution of the cosets of a linear code}\label{subsec2:WD}
 \begin{theorem}\label{th2:HP_coset}
\begin{description}
  \item[(i)] \cite[Lemma 7.5.1]{HufPless}  For a code $\C$, the weight distributions of all cosets $\alpha\vf+\C$, $\alpha\in\F_q^*$, are identical.

  \item[(ii)] \cite[Theorem 7.5.2]{HufPless}, \cite[Theorem 6.20]{MWS}, \cite[Theorem 10.10]{HandbookCodes}
  For a code $\C$, the weight distribution of any coset of weight $<s(\C^\bot)$ is uniquely determined if, in the coset, the numbers
of vectors of weights $1,2,\ldots,s(\C^\bot)-1$ are known.

\item[(iii)] \cite{Bonneau1990,Delsarte4Fundam,HufPless,MWS} For an  MDS code of distance $d$,  the weight distribution of a coset is uniquely determined if, in the coset, the numbers
of vectors of weights $1,2,\ldots,d-2$ are known.
\end{description}
\end{theorem}

Theorem \ref{th2:BonTrans} is obtained in \cite{DMP_MDScosets} by the transformation of the formula from \cite{Bonneau1990}.

\begin{theorem}\label{th2:BonTrans}
\cite{DMP_MDScosets} \textbf{(Bonneau formula transformed)}
Let $\C$ be an $[n,k,d]_q$ MDS code of distance $d\ge2$. Let $\V$ be one of its cosets. Let $A_w(\C)$ be the number of codewords of weight $w$ of $\C$ as in \eqref{eq2:Aw(C)}. Let $B_w(\V)$ be the number of vectors of weight $w$ in the coset~$\V$. Assume that all values of $B_v(\V)$ with $0\le v\le d-2$ are known.  Then, for $w\ge d-1$, the weight distribution of\/ $\V$ is as follows:
\begin{align}\label{eq2:BonGenTransf}
&B_w(\V)=A_w(\C)-\Omega_w^{(0)}(n,d)+\sum_{v=0}^{d-2}\Omega_w^{(v)}(n,d)B_v(\V),~w=d-1,d,\ldots,n,\db\\
&\label{eq2:OmegaDef}
\Omega_w^{(v)}(n,d)=(-1)^{w-d}\binom{n-v}{w-v}\binom{w-1-v}{d-2-v},~d\ge2.
\end{align}
\end{theorem}

\begin{theorem}\label{th2:wd1cosetMDS}
\cite[Theorem 4.1]{DMP_MDScosets}
 Let $\C$ be an $[n,k,d]_q$ MDS code with $d\ge3$. Let $A_w(\C)$ and $\Omega_w^{(v)}(n,d)$ be as in \eqref{eq2:Aw(C)} and \eqref{eq2:OmegaDef}, respectively. Then all its $n(q-1)$ cosets $\V^{(1)}$ of weight $1$ have the same weight distribution $B_w(\V^{(1)})$ of the form:
  \begin{align}\label{eq2:wd1wc}
&B_w(\V^{(1)})=0 \text{ if }  w\in\{0,1,\ldots,d-2\}\setminus\{1\};B_1(\V^{(1)})=1,~B_{d-1}(\V^{(1)})=\binom{n-1}{d-1};\db\\
&B_w(\V^{(1)})=A_{w}(\C)-\Omega_w^{(0)}(n,d)+\Omega_w^{(1)}(n,d)\text{ if }  w=d,d+1,\ldots,n.\notag
\end{align}
\end{theorem}

\begin{theorem}\label{th2:W=2}\cite[Section 5]{DMP_MDScosets}
Let $\C$ be an $[n,k,d]_q$ MDS code of distance $d\ge5$.
Let $\V^{(2)}$ be one of its cosets of weight $2$. Let $A_w(\C)$ and $\Omega_w^{(v)}(n,d)$ be as in \eqref{eq2:Aw(C)} and \eqref{eq2:OmegaDef}.
\begin{description}
  \item[(i)] Assume that  the value of $B_{d-2}(\V^{(2)})$ is known. Then the number $B_w(\V^{(2)})$ of vectors of weight $w$ in the coset $\V^{(2)}$  has the  form:
\begin{align}\label{eq2:Bon_w=2}
&B_w(\V^{(2)})=0 \text{ if } w\in\{0,1,\ldots,d-3\}\setminus\{2\},~B_2(\V^{(2)})=1,\db\\
&B_w(\V^{(2)})=A_w(\C)-\Omega_w^{(0)}(n,d)+\Omega_w^{(2)}(n,d)+(-1)^{w-d}\binom{n-d+2}{n-w}B_{d-2}(\V^{(2)})\db\notag\\
&\text{ if }w= d-1,d,\ldots,n.\notag
\end{align}

  \item[(ii)] The overall number $\B_{d-2}^{\Sigma}(\V^{(2)})$ of vectors of weight $d-2$ in all cosets of weight $2$ of $\C$ is as follows:
\begin{equation}\label{eq5:Bd-2Sum}
   \B_{d-2}^{\Sigma}(\V^{(2)})=(q-1)\binom{n}{2}\binom{n-2}{d-2}.
   \end{equation}

  \item[(iii)] Assume that  all $\binom{n}{2}(q-1)^2$ cosets $\V^{(2)}$ of weight $2$ of $\C$ have the same weight distribution.
Then, the weight distribution  of any coset $\V^{(2)}$ of weight $2$ is as follows:
  \begin{align}\label{eq2:wd2_ident}
&B_w(\V^{(2)})=0 \text{ if } w\in\{0,1,\ldots,d-3\}\setminus\{2\},~B_2(\V^{(2)})=1,\db\\
&B_{d-2}(\V^{(2)})=\frac{\B_{d-2}^{\Sigma}(\V^{(2)})}{\binom{n}{2}(q-1)^2}=\frac{1}{q-1}\binom{n-2}{d-2},\db\notag\\
&B_w(\V^{(2)})
=A_w(\C)-\Omega_w^{(0)}(n,d)+\Omega_w^{(2)}(n,d)
+(-1)^{w-d}\,\binom{n-d+2}{n-w}\cdot\frac{1}{q-1}\binom{n-2}{d-2}\db\notag\\
&\text{if }w=d-1,d,\ldots,n.\notag
\end{align}

  \item[(iv)]  A necessary condition for equality of the weight distributions of all cosets of weight $2$ of $\C$ is as follows:
\begin{equation}\label{eq2:necessar}
\frac{1}{q-1}\binom{n-2}{d-2} \text{ is an integer}.
\end{equation}

  \item[(v)]
    Let $n=q+1$, $d\ge5$. Let $q-1$ be coprime with $d-2$. Then the necessary condition \eqref{eq2:necessar} holds.

  \item[(vi)] \textbf{(Symmetry of different weight distributions)}
    Let $\C$ be an $[n,k,d]_q$ MDS code with $d\ge5$. Let $\V^{(2)}_a$ and $\V^{(2)}_b$ be two of its cosets of weight $2$ with  different weight distributions. Then independently of values of $B_{d-2}(\V^{(2)}_a)$ and $B_{d-2}(\V^{(2)}_b)$, there is the following symmetry of the weight distributions:
\begin{align}
&(-1)^{n+d}B_w(\V^{(2)}_a)-B_{n+d-2-w}(\V^{(2)}_a)=(-1)^{n+d}B_w(\V^{(2)}_b)-B_{n+d-2-w}(\V^{(2)}_b),\notag\\
&w=d-1,d,\ldots,n.
\end{align}
\end{description}
\end{theorem}

\subsection{Completely regular and $t$-regular codes}\label{subsec24:reguar}
Completely regular and $t$-regular codes are considered e.g.\ in
\cite{BorgRiZin2019PIT,BorgRiZin2025arX,Delsarte4Fundam,GoethTilb1975UP}, see also the references therein.
Slightly rephrasing the definitions from the literature, we give the following definition.
 \begin{definition}\label{def2:tregul}
 We consider an $[n,k]_R$ code $\C$. Let $t$ be an integer, $0\le t\le R$. If for \emph{every} $W$ such that $0
 \le W\le t$, all cosets of weight $W$ of $\C$ have the same weight distribution then $\C$ is called a $t$-regular code. Moreover, an $R$-regular code is called \emph{completely regular}.
 \end{definition}

\section{The number of vectors of weight $d-2$ in the cosets of weight 2 of a $[q+1,q+2-d,d]_q$ GDRS code with $d\ge5$; Open
Problems $A_{q,\mu}^\times$ and $A_{\Rk,\mu}^+$}\label{sec:wd-2vectors}

\begin{notation}\label{not2}
 In addition to Notation \ref{not1}, for an $[n,k,d]_{q}$ code $\C$ and its cosets $\V$ \eqref{eq1:coset}, we use the following notations and definitions, where $1\le j_i\le n$, $j_u\ne j_v$ if $u\ne v$:
 \begin{align*}
 &\vf_w&&\text{a vector of weight $w$ of }\F_q^{n},~wt(\vf_w)=w;\db\\
 &\vf_2(j_1,j_2;\gamma_1,\gamma_2)&&\text{a vector of weight 2 of $\F_q^{n}$ with elements }\gamma_1,\gamma_2\in\F_q^*\db\\
 &&&\text{in positions $j_1,j_2$, respectively;}\db\\
 &\cf_d(j_1,j_2,\ldots,j_d)&&\text{a codeword of weight $d$ of $\C$ with non-zero elements}\db\\
 &&&\text{in positions }j_1,j_2,\ldots,j_d;\db\\
 &\cf_d(j_1,j_2,\ldots,j_d;\gamma_1,\gamma_2)&&\text{a codeword of weight $d$ of $\C$ with non-zero elements in }\db\\
 &&&\text{positions $j_1,j_2,\ldots,j_d$ where elements }\gamma_1,\gamma_2\in\F_q^*\db\\
 &&&\text{are placed in positions $j_1,j_2$, respectively;}\db\\
 &\cf_d^*(\vf_2)&&\text{a codeword of weight $d$ of $\C$ \emph{critical} for the vector $\vf_2$ so that}\db\\
 &&&\cf_d^*(\vf_2)+\vf_2=\vf_{d-2};\db\\
 &\cf_d^*(\vf_2(j_1,j_2;\gamma_1,\gamma_2))&&\text{a codeword of weight $d$ of $\C$ \emph{critical} for the vector }\vf_2(j_1,j_2;\gamma_1,\gamma_2)\db\\
 &&&\text{so that }\cf_d^*(\vf_2(j_1,j_2;\gamma_1,\gamma_2))+\vf_2(\vf_2(j_1,j_2;\gamma_1,\gamma_2))=\vf_{d-2}\db\\
 &&&\text{that implies }\cf_d^*(\vf_2(j_1,j_2;\gamma_1,\gamma_2))=\cf_d(j_1,j_2,\ldots,j_d;-\gamma_1,-\gamma_2).
 \end{align*}
\end{notation}
From \eqref{eq1:coset} and Notations \ref{not1}, \ref{not2}, the following lemma holds:
\begin{lemma}\label{lem3:Bd-2=crit}
 The number $B_{d-2}(\V^{(2)})$ of vectors of weight $d-2$ in a coset of weight $2$ with a coset leader $\vf_2$ is equal to the number of the critical codewords $\cf_d^*(\vf_2)$ for $\vf_2$.
\end{lemma}

To investigate the weight distributions of the cosets of the $[q+1,q+2-d,d]_q$ GDRS code, we consider the \emph{normalized} code given by the parity check matrix $\mathbf{H}_d$ of  \eqref{eq2:HDRS} with $v_1=\ldots,v_{q+1}=1$.
Change of elements $v_i$ in  \eqref{eq2:HDRS} does not affect the weight distributions of the cosets.

\begin{definition}\label{def3:Gamma}
 Let $q$ be a prime power. Let $\mu$ be a positive integer, $\mu<q-1$. We introduce the following notion: \emph{$\mu$-factors peculiarity} (\emph{$\mu^\times$-peculiarity}, for short) of an element $\gamma$ of $\F_q^*$ is the number of all possible $\mu$-tuples $\{\alpha_1,\alpha_2,\ldots,\alpha_{\mu}\}$, each of which consists of $\mu$ distinct elements $\alpha_j$ of $\F_q^*$ such that their product is equal to $\gamma$. We denote the $\mu^\times$-peculiarity of $\gamma\in\F_q^*$ by
 $\Pb_{q,\mu}^\times(\gamma)$.  Also, we denote the set of all the mentioned $\mu$-tuples by $\widehat{\Pb}_{q,\mu}^\times(\gamma)$. In other words,
 \begin{align}
&\widehat{\Pb}_{q,\mu}^\times(\gamma)\triangleq
\{\{\alpha_1,\ldots,\alpha_{\mu}\}\subset\F_q^*\,|\, ~\alpha_i\ne \alpha_j\text{ if }i\ne j,~\gamma=\prod_{j=1}^{\mu}\alpha_j\},~\gamma\in\F_q^*,~\mu<q-1;\label{eq3:hatPb gamma}\db\\
&\Pb_{q,\mu}^\times(\gamma)\triangleq \#\widehat{\Pb}_{q,\mu}^\times(\gamma). \label{eq3:Pb gamma}
 \end{align}
\end{definition}

\begin{theorem}\label{th3:Bd-2}
 Let $\C$ be a $[q+1,q+2-d,d]_q$ normalized GDRS code of distance $d\ge5$. Let $\V^{(2)}$ be one of its coset of weight $2$ with a leader $\vf_2(j_1,j_2;\gamma_1,\gamma_2)$.  Then the number $B_{d-2}(\V^{(2)})$ of vectors of weight $d-2$ in $\V^{(2)}$ does not depend on the positions  $j_1,j_2$ and is equal to $(d-2)^\times$-peculiarity of the element $-\gamma_2/\gamma_1$ of $\F_q^*$. In other words,
  \begin{equation}\label{eq3:Bd-2=Gamma}
  B_{d-2}(\V^{(2)})=\Pb^\times_{q,d-2}(-\gamma_2/\gamma_1).
 \end{equation}
 \end{theorem}

\begin{proof}
We find the number of the critical codewords $\cf_d^*(\vf_2(j_1,j_2;\gamma_1,\gamma_2))$ for the coset leaders $\vf_2(j_1,j_2;\gamma_1,\gamma_2)$ with distinct places $j_1,j_2$ and values of $\gamma_1,\gamma_2$. By Lemma \ref{lem3:Bd-2=crit}, this number is $B_{d-2}(\V^{(2)})$ in the corresponding cosets.

In a codeword and in a $(q+1)$-vector of a code coset, the $j$-th position with $1\le j\le q$ corresponds to the column $[1,m_{j},m_{j}^2,\ldots,m_{j}^{d-2}]^{tr}$ of the matrix $\mathbf{H}_d$ of \eqref{eq2:HDRS}, whereas the $(q+1)$-st position corresponds to the column $[0,0,\ldots,0,1]^{tr}$ of $\mathbf{H}_d$.

If $\cf_d$ is a codeword and $\alpha\in\F_q^*$ then $\alpha\cf_d$ is a codeword also. Therefore, for counting of the number of the critical codewords $\cf_d^*(\vf_2)$ one may take an arbitrary value of $\gamma_1$, for instance $\gamma_1=1$. In this case, the number of the critical codewords $\cf_d^*(\vf_2)$ depends on $\gamma_2$. In other words, instead of a leader $\vf_2(j_1,j_2;\gamma_1,\gamma_2)$ we may consider $\vf_2(j_1,j_2;1,\gamma_2/\gamma_1)=\vf_2(j_1,j_2;1,\gamma)$ where $\gamma\in\F_q^*$.

We assume that in all formulas in the proof, we have
\begin{equation}\label{eq3:neq}
 a_i\in\F_q,~a_i\ne a_j\text{ if }i\ne j.
\end{equation}

To find all critical codewords and to count the number of them, we consider convenient systems of $d-1$ linear equations over $\F_q$. Solutions of the systems give the needed critical codewords.

The leaders $\vf_2(j_1,j_2;1,\gamma)$, $\gamma\in\F_q^*$, can be partitioned into two sets with $j_1=q+1$, $1\le j_2\le q$,
for the 1-st set, and  $1\le j_1,j_2\le q$,
 for the 2-nd one.

\textbf{(1)} Let a coset leader have the form $\vf_2(j_1,j_2;1,\gamma)$ with $j_1=q+1$, $1\le j_2\le q$, $\gamma\in\F_q^*$.

We consider the following system of $d-1$ linear equations:
\begin{align}
&\mathbf{A}\times[x_1,x_2,\ldots,x_{d-1}]^{tr}=[0,\ldots,0,1]^{tr},\db\notag\\
&\mathbf{A}=\left[ \begin{array}{cccccc}
        1     &|&1     &1     &\ldots &1      \\
        a_1   &|&a_2   &a_3   &\ldots &a_{d-1}\smallskip    \\
        a_1^2 &|&a_2^2 &a_3^2 & \ldots&a_{d-1}^2\smallskip  \\
        \ldots&|&\ldots&\ldots&\ldots&\ldots\\
        a_1^{d-2} &|&a_2^{d-2} &a_3^{d-2} &\ldots &a_{d-1}^{d-2}  \\
       \end{array}\right],\label{eq3:syst_matrA}
\end{align}
where $\mathbf{A}$ is a submatrix of the matrix $\mathbf{H}_d$ of \eqref{eq2:HDRS} such that $a_1=m_{j_2}$.

The columns $[0,\ldots,0,1]^{tr}$ and $[1,a_1,a_1^2,\ldots,a_1^{d-2}]^{tr}$ correspond, respectively, to the positions $j_1=q+1$ and $j_2$ of the coset leader $\vf_2(j_1,j_2;1,\gamma)$.

As $\mathbf{A}$ is a Vandermond matrix, by Cramer rule, we have
\begin{equation}\label{eq3:x1}
x_1=\frac{(-1)^{d-1+1}\prod\limits_{2\le i<j\le d-1}(a_j-a_i)}{\prod\limits_{1\le i<j\le d-1}(a_j-a_i)}=
(-1)^{d}\prod_{j=2}^{d-1}\frac{1}{a_j -a_1}=\prod_{j=2}^{d-1}\frac{1}{a_1-a_j}.
\end{equation}
We fix $a_1=m_{j_2}$ and put that $a_2,\ldots,a_{d-1}$ are chosen so that
 $x_1=-\gamma$. Let $a_u=m_{j_{u+1}}$, $u=2,3,\ldots,d-1$, where $j_{u+1}$ are non-zero positions of a critical codeword.
Then $\cf_d=$\linebreak
 $(q+1,j_2,j_3,\ldots,j_d;-1,-\gamma)$ is the critical codeword for the leader $\vf_2(j_1,j_2;1,\gamma)$.

Let $M_q(a_1,\gamma)$ be the number of distinct $(d-2)$-sets $\{a_2,a_3,\ldots,a_{d-1}\}$ satisfying~\eqref{eq3:x1} with  the fixed $a_1=m_{j_2}$, under the conditions $x_1=-\gamma$ and \eqref{eq3:neq}.

We denote
 $\widetilde{a}_j=(a_1-a_{j+1})^{-1},~j=1,\ldots,d-2$,
that implies, by \eqref{eq3:x1},
\begin{equation}\label{eq3:x1transf}
-\gamma=\prod_{j=1}^{d-2}\widetilde{a}_j,~\gamma, \widetilde{a}_j\in\F_q^*,~\widetilde{a}_u\ne \widetilde{a}_v\text{ if }u\ne v.
\end{equation}
We denote by $\widetilde{M}_q(\gamma)$ the number of distinct $(d-2)$-sets $\{\widetilde{a}_1,\widetilde{a}_2,\ldots,\widetilde{a}_{d-2}\}$ satisfying \eqref{eq3:x1transf}. By above, $\widetilde{M}_q(\gamma)=M_q(a_1,\gamma)$. So, $M_q(a_1,\gamma)$ does not depend on~$a_1$. By \eqref{eq3:hatPb gamma},
\eqref{eq3:Pb gamma}, \eqref{eq3:x1transf},
\begin{equation*}
\widetilde{M}_q(\gamma)=\Pb^\times_{q,d-2}(-\gamma).
\end{equation*}

\textbf{(2)} Let a coset leader has the form $\vf_2(j_1,j_2;1,\gamma)$, $1\le j_1,j_2\le q$, $\gamma\in\F_q^*$.

The critical codewords $\cf_d^*(\vf_2)$  of the leader can be partitioned into two sets of the following forms:
 $\cf_d(j_1,j_2,\ldots,j_d)$, $j_i\le q$, $i=1,2,\ldots,d$,
for the 1-st set, and\\
 $\cf_d(j_1,j_2,\ldots,j_{d-1},j_d)$, $j_i\le q$, $i=1,\ldots,d-1$, $j_d=q+1$, for the 2-nd one.

\textbf{(2.1)} Let $\cf_d^*(\vf_2)$ have the form $\cf_d(j_1,j_2,\ldots,j_d)$ with $j_i\le q$, $i=1,2,\ldots,d$.

We consider the system of $d-1$ linear equations\\
 $\mathbf{A}\times[y_1,y_2,\ldots,y_{d-1}]^{tr}=[1,a_0,a_0^2,\ldots,a_0^{d-2}]^{tr},$
where $\mathbf{A}$  is as in \eqref{eq3:syst_matrA} with $a_1=m_{j_2}$.

We put $a_0=m_{j_1}$. Then the columns $[1,a_0,a_0^2,\ldots,a_0^{d-2}]^{tr}$ and $[1,a_1,a_1^2,\ldots,a_1^{d-2}]^{tr}$ correspond, respectively, to the positions $j_1$ and $j_2$ of the coset leader.

 By Cramer rule, we have
\begin{equation}\label{eq3:y1}
y_1=\frac{\prod\limits_{0\le i<j\le d-1,\,i,j\ne1}(a_j-a_i)}{\prod\limits_{1\le i<j\le d-1}(a_j-a_i)}=
\prod_{j=2}^{d-1}\frac{a_j-a_0}{a_j-a_1}=\prod_{j=2}^{d-1}\left(1+\frac{a_1-a_0}{a_{j}-a_1}\right).
\end{equation}
We fix $a_0=m_{j_1}$, $a_1=m_{j_2}$ and put that $a_2,\ldots,a_{d-1}$ are chosen so that
 $y_1=-\gamma$. Let $a_u=m_{j_{u+1}}$, $u=2.\ldots,d-1$, where $j_{u+1}$ are non-zero positions of a critical codeword.
Then $\cf_d=(j_1,j_2,\ldots,j_d;-1,-\gamma)$ is the critical codeword for the leader $\vf_2(j_1,j_2;1,\gamma)$.

Let $T_q(a_0,a_1,\gamma)$ be the number of distinct $(d-2)$-sets $\{a_2,a_3,\ldots,a_{d-1}\}$ satisfying \eqref{eq3:y1} with the fixed $a_0=m_{j_1}$, $a_1=m_{j_2}$, under the conditions $y_1=-\gamma$ and \eqref{eq3:neq}.

We denote
\begin{equation} \label{eq3:widehat_notat}
\widehat{a}_j=1+\frac{a_1-a_0}{a_{j+1}-a_1},~j=1,\ldots,d-2.
\end{equation}
By  \eqref{eq3:y1}, \eqref{eq3:widehat_notat}, we have
\begin{equation}\label{eq3:y1transf}
-\gamma=\prod_{j=1}^{d-2}\widehat{a}_j,
\gamma\in\F_q^*,~ \widehat{a}_j\in\F_q^*\setminus\{1\},~\widehat{a}_u\ne \widehat{a}_v\text{ if }u\ne v.
\end{equation}
Denote by $\widehat{T}_q(\gamma)$ the number of distinct $(d-2)$-sets $\{\widehat{a}_1,\widehat{a}_2,\ldots,\widehat{a}_{d-2}\}$ satisfying~\eqref{eq3:y1transf}.
Obviously, $\widehat{T}_q(\gamma)=T_q(a_0,a_1,\gamma)$. In particular, this means that $T_q(a_0,a_1,\gamma)$ does not depend on $a_0,a_1$.

\textbf{(2.2)} Let $\cf_d^*(\vf_2)$ have the form
 $\cf_d(j_1,\ldots,j_{d-1},j_d)$, $j_i\le q$, $i=1,\ldots,d-1$, $j_d=q+1$.

We consider the system  of $d-1$ linear equations
 \begin{align*}
 &\mathbf{A}'[z_1,z_2,\ldots,z_{d-1}]^{tr}=[1,a_0,a_0^2,\ldots,a_0^{d-2}]^{tr},\db\\
&\mathbf{A}'=\left[ \begin{array}{ccccccc}
        1     &|&1     &1     &\ldots &1            &0\\
        a_1   &|&a_2   &a_3   &\ldots &a_{d-2}      &0\smallskip    \\
         \ldots&|&\ldots&\ldots&\ldots&\ldots        &\ldots\\
        a_1^{d-3} &|&a_2^{d-3} &a_3^{d-3} &\ldots &a_{d-2}^{d-3}&0\smallskip   \\
        a_1^{d-2} &|&a_2^{d-2} &a_3^{d-2} &\ldots &a_{d-2}^{d-2}&1  \\
       \end{array}\right],
 \end{align*}
 where $\mathbf{A}'$ is a submatrix of the matrix $\mathbf{H}_d$ of \eqref{eq2:HDRS} such that $a_1=m_{j_2}$.

The columns $[1,a_0,a_0^2,\ldots,a_0^{d-2}]^{tr}$ and $[1,a_1,a_1^2,\ldots,a_1^{d-2}]^{tr}$, where $a_0=m_{j_1}$ and $a_1=m_{j_2}$, correspond, respectively, to the positions $j_1$ and $j_2$ of the   coset leader $\vf_2(j_1,j_2;1,\gamma)$.

By Cramer rule,
\begin{equation}\label{eq3:z1_2}
z_1=\frac{(-1)^{d-1+d-1}\prod\limits_{0\le i<j\le d-2,\,i,j\ne1}(a_j-a_i)}{(-1)^{d-1+d-1}\prod\limits_{1\le i<j\le d-2}(a_j-a_i)}=\prod\limits_{j=2}^{d-2}\frac{a_j-a_0}{a_j-a_1}=\prod_{j=2}^{d-2}\left(1+\frac{a_1-a_0}{a_{j}-a_1}\right).
\end{equation}
We fix $a_0=m_{j_1}$, $a_1=m_{j_2}$, and put that $a_2,\ldots,a_{d-1}$ are chosen so that
 $z_1=-\gamma$. Let $a_u=m_{j_{u+1}}$, $u=2.\ldots,d-1$, where $j_{u+1}$ are non-zero positions of a critical codeword.
Then $\cf_d=(j_1,j_2,\ldots,j_d;-1,-\gamma;)$ is the critical codeword for the leader $\vf_2(j_1,j_2;1,\gamma)$.

Let $T'_q(a_0,a_1,\gamma)$ be the number of distinct $(d-3)$-sets $\{a_2,a_3,\ldots,a_{d-2}\}$ satisfying~\eqref{eq3:z1_2} with the fixed $a_0=m_{j_1}$, $a_1=m_{j_2}$, under the conditions $z_1=-\gamma$ and \eqref{eq3:neq}. Using the notation \eqref{eq3:widehat_notat}, similarly to \eqref{eq3:y1transf}, we obtain
\begin{equation}
-\gamma=\prod_{j=1}^{d-3}\widehat{a}_j,~\gamma\in\F_q^*,~ \widehat{a}_j\in\F_q^*\setminus\{1\},~\widehat{a}_u\ne \widehat{a}_v\text{ if }u\ne v.\label{eq3:z1transf}
\end{equation}
Denote by $\widehat{T'}_q(\gamma)$ the number of distinct $(d-3)$-sets $\{\widehat{a}_1,\widehat{a}_2,\ldots,\widehat{a}_{d-3}\}$ satisfying~\eqref{eq3:z1transf}.
Obviously, $\widehat{T'}_q(\gamma)=T'_q(a_0,a_1,\gamma)$. In particular, this means that $T'_q(a_0,a_1,\gamma)$ does not depend on $a_0,a_1$.

 By the cases  (2.1) and (2.2), the total number of the critical  codewords $\cf_d^*(\vf_2)$ of the form
 $\cf_d(j_1,\ldots,j_{d-1},j_d)$, $j_i\le q$, $i=1,\ldots,d-1$, $j_d=q+1$, is  $ \widehat{T}_q(\gamma) +\widehat{T'}_q(\gamma) $. The factor $\widehat{a}_{d-2}=1$ does not change the product in \eqref{eq3:z1transf}. Therefore,
 $$ \widehat{T}_q(\gamma) +\widehat{T'}_q(\gamma) =\Pb^\times_{q,d-2}(-\gamma),$$
 see \eqref{eq3:hatPb gamma}, \eqref{eq3:Pb gamma}. It remains to remind $\widetilde{M}_q(\gamma)=\Pb^\times_{q,d-2}(-\gamma)$, see the case (1).
 \end{proof}

Let $\Z_\Rk$ be the ring of integers modulo $\Rk$; see, e.g., \cite{IreRosNumbTheor} for introduction to this topic.

\begin{definition}\label{def3:Lambda}
Let $\mu$ and $\Rk$ be positive integers, $\mu<\Rk$. We introduce the following notion: $\mu$-\emph{summands peculiarity} ($\mu^+$-\emph{peculiarity}, for short) of an element $\lambda$ of the ring $\Z_{\Rk}$ is the number of all possible $\mu$-tuples $\{\lambda_1,\lambda_2,\ldots,\lambda_{\mu}\}$, each of which consists of $\mu$ distinct  elements $\lambda_j$ of $\Z_{\Rk}$ such that their sum in $\Z_{\Rk}$ is equal to $\lambda$. We denote the $\mu^+$-peculiarity of $\lambda\in\Z_{\Rk}$ by $\Ps_{\Rk,\mu}^+(\lambda)$. Also, we denote the set of all the mentioned $\mu$-tuples by $\widehat{\Ps}_{q,\mu}^+(\gamma)$. In other words,
  \begin{align}
&\widehat{\Ps}^+_{\Rk,\mu}(\lambda)\triangleq\left\{\{\lambda_1,\ldots,\lambda_\mu\}\subset\Z_\Rk\,|\,\lambda_i\ne\lambda_j\text{ if }i\ne j, \sum_{j=1}^\mu\lambda_j=\lambda\right\},\lambda\in\Z_{\Rk},\mu<\Rk.\label{eq3:hatPk lambda}\db\\
&\Ps_{\Rk,\mu}^+(\lambda)\triangleq\#\widehat{\Ps}_{q,\mu}^+(\lambda).\label{eq3:Pk lambda}
\end{align}
For $\widehat{\Ps}^+_{\Rk,\mu}(\lambda)$ and $\Ps^+_{\Rk,\mu}(\lambda)$ we use also the notations $\widehat{\Ps}^+_{\Rk,\mu}(\lambda)_{\lambda_1,\ldots,\lambda_\mu}$ and $\Ps^+_{\Rk,\mu}(\lambda)_{\lambda_1,\ldots,\lambda_\mu}$.
\end{definition}

\begin{lemma}\label{lem3:Gamma=Lamb}
Let $\beta_q$ be a primitive element of $\F_q$. We have,
\begin{equation}\label{eq3:Gamma=Lamb}
\Pb^\times_{q,\mu}(\beta_q^\lambda)=\Ps^+_{q-1,\mu}(\lambda),~\lambda\in\{0,1,\ldots,q-2\}.
\end{equation}
\end{lemma}

\begin{proof}
We remind that in $\F_q$, the logarithms are calculated modulo $q-1$. In \eqref{eq3:hatPb gamma}, \eqref{eq3:Pb gamma}, we put $\gamma=\beta_q^\lambda$, $\alpha_j=\beta_q^{\lambda_j}$, $\lambda_j\in\{0,1,\ldots,q-2\}$, that implies
 \begin{equation*}
 \beta_q^\lambda=\prod_{j=1}^{\mu}\beta_q^{\lambda_j}, ~\lambda=\left(\sum_{j=1}^{\mu}\lambda_j\right)\pmod{q-1},~\lambda,\lambda_j\in\Z_{q-1},~\lambda_u\ne \lambda_v\text{ if }u\ne v.
 \qedhere
\end{equation*}
\end{proof}

\begin{corollary}\label{cor3:Bd-2}
 Let $\C$ be a $[q+1,q+2-d,d]_q$ normalized GDRS code of distance $d\ge5$. Let $\V^{(2)}$ be one of its coset of weight $2$ with a leader $\vf_2(j_1,j_2;\gamma_1,\gamma_2)$. Then the number $B_{d-2}(\V^{(2)})$ of vectors of weight $d-2$ in $\V^{(2)}$ does not depend on the positions  $j_1,j_2$ and is as follows
  \begin{align}\label{eq3:Bd-2=Lamb}
&B_{d-2}(\V^{(2)})=\Ps^+_{q-1,d-2}(\lambda(\gamma_1,\gamma_2)),~\lambda(\gamma_1,\gamma_2)\in~\Z_{q-1},~
  \beta_q^{\lambda(\gamma_1,\gamma_2)}=-\frac{\gamma_2}{\gamma_1},\db\\
&\beta_q\text{ is a primitive element of }\F_q.\notag
 \end{align}
 \end{corollary}

 \begin{proof}
   The assertion follows from Theorem \ref{th3:Bd-2} and Lemma \ref{lem3:Gamma=Lamb}.
 \end{proof}

\begin{lemma}\label{lem3:gam&gamdeg-1 lamb&-lamb}
  Let $\Pb^\times_{q,\mu}(\gamma)$ and $\Ps^+_{\Rk,\mu}(\lambda)$ be as in Definitions \ref{def3:Gamma} and \ref{def3:Lambda}. Then
 \begin{equation}\label{eq3:gam&gamdeg-1 lamb&-lamb}
 \Pb^\times_{q,\mu}(\gamma)=\Pb^\times_{q,\mu}(\gamma^{-1}),~\Ps^+_{\Rk,\mu}(\lambda)=\Ps^+_{\Rk,\mu}(-\lambda).
 \end{equation}
\end{lemma}

\begin{proof}
  The assertions follow from \eqref{eq3:hatPb gamma}--\eqref{eq3:Bd-2=Gamma}, \eqref{eq3:hatPk lambda}--\eqref{eq3:Bd-2=Lamb}.
\end{proof}

\begin{problem}\label{openprobAx}
  (\textbf{Open Problem $A_{q,\mu}^\times$ for Galois fields $\F_q$}) Let $\gamma$ be an element of the multiplicative group $\F_q^*$ of the field $\F_q$. Let the $\mu^\times$-peculiarity $\Pb^\times_{q,\mu}(\gamma)$ of $\gamma$ be as in \eqref{eq3:hatPb gamma}, \eqref{eq3:Pb gamma} of Definition \ref{def3:Gamma}. The problem is to find values of $\Pb^\times_{q,\mu}(\gamma)$ for all possible triples $q,\mu,\gamma$.
\end{problem}

\begin{problem}\label{openprobA+}
  (\textbf{Open Problem $A_{\Rk,\mu}^+$ for rings $\Z_\Rk$ of integers modulo $\Rk$}) Let $\lambda$ be an element of the ring $\Z_\Rk$ of integers modulo $\Rk$. Let the $\mu^+$-peculiarity $\Ps^+_{\Rk,\mu}(\lambda)$ of $\lambda$ be as in \eqref{eq3:hatPk lambda} \eqref{eq3:Pk lambda} of Definition \ref{def3:Lambda}. The problem is to find values of $\Ps^+_{\Rk,\mu}(\lambda)$ for all possible triples $\Rk,\mu,\lambda$ or, at least, when in the triple, $\Rk+1$ is a prime power.
\end{problem}

\section{Partial solution of Problem $A_{\Rk,\mu}^+$ for the rings $\Z_\Rk$  of integers modulo $\Rk$; orbits of elements $\lambda$ of the rings and values of $\Ps^+_{\Rk,\mu}(\lambda)$}
\label{sec4:orbits}
 We use Definition \ref{def3:Lambda} where $\Z_{\Rk}$ is the ring of integers modulo $\Rk$. Let $\Z_\Rk^*$ be the multiplicative group of integers modulo $\Rk$ which is the group of units of the ring $\Z_\Rk$; it consists of integers coprime with $\Rk$.
We have $\#\Z_\Rk^*=\phi(\Rk)$, where $\phi(\Rk)$ is Euler's totient function.
Let $(a,b)$ be the greatest common divisor of the integers $a$ and $b$.
The results of this section hold for any $\Rk$, including cases when $\Rk+1$ is not a prime power.

\subsection{Maps and orbits in rings $\Z_{\Rk}$}\label{subsec41}

\begin{lemma}\label{lem4:biject}
There are the following bijections  between the sets $\widehat{\Ps}^+_{\Rk,\mu}(\lambda)$:
\begin{align}
 &\widehat{\Ps}^+_{\Rk,\mu}(\lambda)_{\lambda_1,\ldots,\lambda_\mu}\Leftrightarrow
\widehat{\Ps}^+_{\Rk,\mu}(\lambda\ell+u\mu)_{\lambda_1\ell+u,\ldots,\lambda_\mu\ell+u},~\lambda,\lambda_j\in\Z_{\Rk},~\ell\in \Z_\Rk^*,
\label{eq4:biject}\db\\
 &u\in\{ 0,1, \dots, \Rk-1\}.\notag
\end{align}
\end{lemma}

\begin{proof}
 As $\sum_{j=1}^\mu\lambda_j=\lambda$, we have $\sum_{j=1}^\mu(\lambda_j\ell+u)=\lambda\ell +u\mu$. Moreover, as $\ell$ is coprime with $\Rk$ and $\lambda_i\ne\lambda_j$ if $i\ne j$,  we have $\lambda_i\ell+u\ne\lambda_j\ell+u\text{ if }i\ne j$.
\end{proof}

We denote by $G_{\Rk,\mu}$ the group of mappings in $\Z_\Rk$ of the form
\begin{equation}\label{eq4:G_mu}
\varphi\in G_{\Rk,\mu}\Rightarrow\lambda\varphi =\lambda\ell+u\mu,~\lambda\in\Z_\Rk,~\ell\in\Z_\Rk^*,~u\in\{ 0,1, \dots, \Rk-1\},~\mu<\Rk.
\end{equation}
The mappings with $\ell=1$ form a subgroup $G_{\Rk,\mu}^\oplus$ of $G_{\Rk,\mu}$.

\begin{definition}\label{def4:orbitsOc&Oplus}
  A subset $\Oc$ of $\Z_\Rk$ is called an \emph{orbit} under $G_{\Rk,\mu}$ if for any $\lambda_1,\lambda_2\in\Oc\subseteq\Z_\Rk$ there is a map $\varphi\in G_{\Rk,\mu}$ such that $\lambda_1\varphi=\lambda_2$. Orbits $\Os^\oplus$ under $G_{\Rk,\mu}^\oplus$ are defined similarly. The notations of the orbit can have a subscript if it is needed by the context. It is possible that an orbit under $G_{\Rk,\mu}^\oplus$ coincides with one under $G_{\Rk,\mu}$.
\end{definition}

We denote
\begin{equation}\label{eq4:DRkmu}
 \Dk_{\Rk,\mu}\triangleq(\Rk,\mu)\in\{1,2,\ldots,\mu\},~\mu<\Rk,~\Phi_\mu\triangleq\frac{\mu}{\Dk_{\Rk,\mu}},
 ~\Phi_{\Rk}\triangleq\frac{\Rk}{\Dk_{\Rk,\mu}}.
\end{equation}
From \eqref{eq4:DRkmu} it follows that
\begin{equation}\label{eq4:DRkmu2}
\Phi_\mu<\Phi_{\Rk},~\Phi_\mu\ge1,~\Phi_{\Rk}\ge2,~(\Phi_\mu,\Phi_\Rk)=1.
\end{equation}

\begin{theorem}\label{th4:orbits}
 For the orbits under $G_{\Rk,\mu}^\oplus$ and $G_{\Rk,\mu}$, the following hold (all calculations are executed modulo $\Rk$):
  \begin{description}
    \item[(i)] The subgroup $G_{\Rk,\mu}^\oplus$ of $G_{\Rk,\mu}$ partitions $\Z_\Rk$ into $\Dk_{\Rk,\mu}$ orbits
 $\Os_j^\oplus$ of the form
    \begin{align}\label{eq4:orbitsRkj}
   &  \Os^\oplus_j=\{j,j+\Dk_{\Rk,\mu},j+2\Dk_{\Rk,\mu},\ldots,j+(\Phi_{\Rk}-1)\Dk_{\Rk,\mu}\}\subseteq\Z_\Rk,
   ~\#\Os^\oplus_j=\Phi_\Rk;\db\\
   &\Os^\oplus_j=\{\lambda\,|\,\lambda\in\Z_\Rk,\,\lambda\equiv j\pmod{\Dk_{\Rk,\mu}}\};~j=0,1,\ldots,\Dk_{\Rk,\mu}-1; ~\bigcup_{j=0}^{\Dk_{\Rk,\mu}-1} \Os^\oplus_j=\Z_\Rk.\notag
    \end{align}
    \item[(ii)]
    For any value of $\mu,\Rk$, and $\Dk_{\Rk,\mu}$, the orbit $\Os^\oplus_0$ is also an orbit under $G_{\Rk,\mu}$, say $\Oc_0$.
    In other words, $\Oc_0=\Os^\oplus_0$.

    \item[(iii)]
    Let there exist $\ell\ne1$, $\ell\in\Z_\Rk^*$, such that $(j_1+u_1\Dk_{\Rk,\mu})\ell=j_2+u_2\Dk_{\Rk,\mu}$. Then and only then, the orbits  $\Os^\oplus_{j_1}$ and $\Os^\oplus_{j_2}$ are subsets of the same orbit under $G_{\Rk,\mu}$.

    \item[(iv)]
    Let $\mu$ be coprime with $\Rk$, i.e.  $\Dk_{\Rk,\mu}=1$ and $\Rk=\Phi_\Rk$. Then $\Oc_0=\Os^\oplus_0=\Z_\Rk$ and $\Oc_0=\Os^\oplus_0$ is the unique orbit.

    \item[(v)]
    Let $\mu$ be not coprime with $\Rk$, i.e.  $\Dk_{\Rk,\mu}\in\{2,3,\ldots,\mu\}$. Then, under $G_{\Rk,\mu}$, there is at least one orbit in addition to   $\Oc_0=\Os^\oplus_0$.

    \item[(vi)]
   Let $\mu$ be not coprime with $\Rk$ and let $\Dk_{\Rk,\mu}=\mu$ be a prime, i.e. $\Rk\equiv0\pmod\mu$, $\Rk=\Phi_\Rk\cdot\mu$, $\Phi_\mu=1$. Then, under $G_{\Rk,\mu}$, there are exactly two orbits $\Oc_0=\Os^\oplus_0$ and
    $\Oc_1=\Z_\Rk\setminus\Oc_0$ such that
\begin{align}\label{eq4:D=muPrimeO0}
&\Oc_0=\Os^\oplus_0=\{0,\mu,\ldots,(\Phi_\Rk-1)\mu\}=
\{\lambda\,|\,\lambda\in\Z_\Rk,\,\lambda\equiv 0\pmod{\Dk_{\Rk,\mu}}\};\db\\
&\Oc_1=\Z_\Rk\setminus\Oc_0=\Z_\Rk\setminus\Os^\oplus_0=\bigcup_{j=1}^{\mu-1}\Os_j^\oplus
=\{\lambda\,|\,\lambda\in\Z_\Rk,\,\lambda\not\equiv 0\pmod{\Dk_{\Rk,\mu}}\}.\label{eq4:D=muPrimeO1}
\end{align}

    \item[(vii)] Let $\Dk_{\Rk,\mu}\in\{2,3,\ldots,\mu\}$ be a prime. Then, under $G_{\Rk,\mu}$, there are exactly two orbits $\Oc_0=\Os^\oplus_0$ and
    $\Oc_1=\Z_\Rk\setminus\Oc_0$, see \eqref{eq4:D=muPrimeO0}, \eqref{eq4:D=muPrimeO1}.
  \end{description}
\end{theorem}

\begin{proof}
   \begin{description}
    \item[(i)] The assertion follows from \eqref{eq4:DRkmu}, \eqref{eq4:DRkmu2}. The relation $(\Phi_\mu,\Phi_{\Rk})=1$ from \eqref{eq4:DRkmu2} is important for the proof.

    \item[(ii)] The map $\lambda\rightarrow\lambda\ell+u\mu$ maps an orbit $\Os_0^\oplus$ to itself. Thus, $\Oc_0=\Os^\oplus_0$.

    \item[(iii)] The assertion follow from the definitions of the groups $\Z_\Rk^*$, $G_{\Rk,\mu}$,  $G_{\Rk,\mu}^\oplus$, and the orbits $\Oc$, $\Os^\oplus$, see \eqref{eq4:G_mu}, Definition \ref{def4:orbitsOc&Oplus}.

    \item[(iv)] As $\mu$ is coprime with $\Rk$, for any fixed $\lambda\in\Z_\Rk$ we have $\Z_\Rk=\{\lambda+u\mu\,|\,~u=0,1,\ldots,\Rk-1\}$, cf. \eqref{eq4:biject}. In other words, $\mu$ generates $\Rk$.

    \item[(v)] The assertion follows from the case (iii).

    \item[(vi)] By the cases (i), (ii) we have $\Dk_{\Rk,\mu}=\mu$ orbits $\Os_j^\oplus$ \eqref{eq4:orbitsRkj} and $\Oc_0=\Os^\oplus_0$.

Let $\ell\ne1$, $\ell\in\Z_\Rk^*$. The map $\lambda\rightarrow\lambda\ell$ maps an orbit $\Os_0^\oplus$ to itself.
Also $\lambda\rightarrow\lambda\ell$ maps an orbit $\Os_j^\oplus$ with $j\ne0$ to $\Os_{\ell j\pmod\mu}^\oplus$ where the subscript $\ell j$ is calculated modulo~$\mu$. Let $\ell'\equiv\ell\pmod\mu$, $1<\ell'<\mu$. As $\mu$ is a prime, $\ell'$ is a primitive element of $\F_\mu$, i.e. the powers of $\ell'$ generate $\mu$. So,  the union of the orbits $\Os_{\ell j\pmod\mu}^\oplus$ forms the orbit $\Oc_1=\bigcup_{j=1}^{\mu-1}\Os_j^\oplus=\Z_\Rk\setminus\Os^\oplus_0=\Z_\Rk\setminus\Oc_0$.

    \item[(vii)] The assertion follows from the cases (v) and (iii). \qedhere
  \end{description}
\end{proof}

\begin{example}\label{exam4:}
  Let $\Rk+1=11$, $\Rk=10$, $\mu=4<\Rk$. Then $\Dk_{\Rk,\mu}=\Dk_{10,4}=2$, $\Phi_\Rk=\frac{10}{2}=5$, $\Phi_\mu=\frac{4}{2}=2$,
$\Os_0^\oplus= \{0,2,4,6,8\},~\Os_1^\oplus=\{1,3,5,7,9\}.$
\end{example}

\subsection{Values of $\Ps^+_{\Rk,\mu}(\lambda)$}
\begin{lemma}\label{lem4:Lambda==}
  Assume that the values of $\Ps^+_{\Rk,\mu}(\lambda)$ are the same for all $\lambda\in\Z_\Rk$. Then
\begin{equation}\label{eq4:Lambda==}
\Ps^+_{\Rk,\mu}(\lambda)=\frac{\binom{\Rk}{\mu}}{\#\Z_\Rk}=\frac{1}{\Rk}\binom{\Rk}{\mu}=\frac{1}{\mu}\binom{\Rk-1}{\mu-1},~\forall\lambda\in\Z_\Rk,
\end{equation}
where $\frac{1}{\mu}\binom{\Rk-1}{\mu-1}$ is an integer.
\end{lemma}

\begin{proof}
  In total, there are $\binom{\Rk}{\mu}$ distinct $\mu$-tuples $\{\lambda_1,\ldots,\lambda_{\mu}\}$ of distinct  elements of  $\Z_\Rk$. Also, $\#\Z_\Rk=\Rk$. Finally, the value  $\Ps^+_{\Rk,\mu}(\lambda)$ must be an integer.
\end{proof}

\begin{theorem}\label{th4:Lambda==}
\begin{description}
  \item[(i)]
  For all elements $\lambda$ of an orbit $\Oc$ under $G_{\Rk,\mu}$, the values of $\Ps_{\Rk,\mu}^+(\lambda)$ are identical. The same holds for any orbit $\Os^\oplus$ under  $G_{\Rk,\mu}^\oplus$.

  \item[(ii)] Let all elements $\lambda$ of $\Z_\Rk$ form a unique orbit under $G_{\Rk,\mu}$. Then $\Ps^+_{\Rk,\mu}(\lambda)=\frac{1}{\mu}\binom{\Rk-1}{\mu-1}$  for all $\lambda\in\Z_\Rk$.
\end{description}
\end{theorem}

\begin{proof}
\begin{description}
  \item[(i)]
The assertion follows from \eqref{eq4:G_mu} and Lemma \ref{lem4:biject}.

  \item[(ii)]
  The assertion follows from Lemma \ref{lem4:Lambda==} with \eqref{eq4:Lambda==} and the case (i) of this theorem.\qedhere
\end{description}
\end{proof}

\begin{theorem}\label{th4:coprimeLamb}
     Let $\mu$ be coprime with $\Rk$, i.e. $(\mu,\Rk)=1$. Then,  $\Ps^+_{\Rk,\mu}(\lambda)=\frac{1}{\mu}\binom{\Rk-1}{\mu-1}$  for all $\lambda\in\Z_\Rk$.
\end{theorem}

\begin{proof}
 We use Theorem \ref{th4:orbits}(iii) and Theorem \ref{th4:Lambda==}(ii).
\end{proof}

\begin{conjecture}\label{conj4:a}
If and only if $\mu$ is coprime with $\Rk$ then values of $\Ps^+_{\Rk,\mu}(\lambda)$ are the same for all $\lambda\in\Z_\Rk$ and the relation \eqref{eq4:Lambda==} holds.
\end{conjecture}

Note that due to Theorem \ref{th4:orbits}(iv), if $\mu$ is not coprime with $\Rk$, we have at least two orbits under $G_{\Rk,\mu}$. However, we do not know if always for elements $\lambda$ from distinct orbits values $\Ps^+_{\Rk,\mu}(\lambda)$ are distinct.

\begin{definition}\label{def4:ch-vector}
\begin{description}
  \item[(i)] Let $\Dk_{\Rk,\mu},\mu,\Rk$, be as in \eqref{eq4:DRkmu}.
 We consider $\mu$-sets $\{\lambda_1,\ldots,\lambda_\mu\}$ from $\widehat{\Ps}^+_{\Rk,\mu}(\lambda)$, see Definition \ref{def3:Lambda} with
 \eqref{eq3:hatPk lambda}, \eqref{eq3:Pk lambda}. We associate each $\mu$-set with a  vector
 \begin{equation}\label{eq4:ch-vectorView}
   C^{(i)}(\lambda)\triangleq(N_0^{(i)}(\lambda),N_1^{(i)}(\lambda),\ldots,N_{\Dk_{\Rk,\mu}-1}^{(i)}(\lambda)),
 \end{equation}
where $N_\xi^{(i)}(\lambda)$ is the number of elements $\lambda_j$ from this set such that $\lambda_j\equiv\xi\pmod{\Dk_{\Rk,\mu}}$. The superscript $(i)$ notes that this vector can be not unique for the given $\lambda$, i.e. for some $\mu$-sets the vectors can be distinct. We can omit the superscript if it is not relevant. In other words
\begin{align}\label{eq4:ch-vectorDetails}
&N_\xi^{(i)}(\lambda)=\#\{\lambda_j\,|\,\lambda_j\equiv\xi\pmod{\Dk_{\Rk,\mu}},\{\lambda_1,\ldots,\lambda_\mu\}
\in\widehat{\Ps}^+_{\Rk,\mu}(\lambda)\},\db\\
&0\le N_\xi^{(i)}(\lambda)\le\mu,~0\le\xi\le\Dk_{\Rk,\mu}-1,~1\le i\le\Ps_{\Rk,\mu}^+(\lambda),
\sum_{\xi=0}^{\Dk_{\Rk,\mu}-1}N_\xi^{(i)}(\lambda)=\mu.\notag
\end{align}
The vector $C^{(i)}(\lambda)=(N_0^{(i)}(\lambda),N_1^{(i)}(\lambda),\ldots,N_{\Dk_{\Rk,\mu}-1}^{(i)}(\lambda))$ is called \emph{the $i$-th characteristic vector of $\lambda$} (\emph{the $i$-th ch$(\lambda)$-vector} or simply \emph{ch$(\lambda)$-vector}, for short).

  \item[(ii)] Let $C_\delta^{(i)}(\lambda)$ be the cyclic shift of a ch$(\lambda)$-vector $C^{(i)}(\lambda)=(N_0^{(i)}(\lambda),\ldots,N_{\Dk_{\Rk,\mu}-1}^{(i)}(\lambda))$ to the right by $\delta$ positions, $\delta=1,2,\ldots,\Dk_{\Rk,\mu}$, i.e.
 \begin{equation}\label{eq4:ch-vectorCyclic}
   C_\delta^{(i)}(\lambda)\triangleq(N_{\Dk_{\Rk,\mu}-\delta}^{(i)}(\lambda),\ldots,
   N_{\Dk_{\Rk,\mu}-1}^{(i)}(\lambda),N_0^{(i)}(\lambda),N_1^{(i)}(\lambda),\ldots,N_{\Dk_{\Rk,\mu}-\delta-1}^{(i)}(\lambda));
 \end{equation}
in particular, $C_{\Dk_{\Rk,\mu}}^{(i)}(\lambda)=C^{(i)}(\lambda)$. We denote the number of distinct cyclic shifts of the vector $C^{(i)}(\lambda)$ by $m^{(i)}(\lambda)$. Let $C^{(i)}_\Sigma(\lambda)$ be the $m^{(i)}(\lambda)$-set of \emph{all distinct} cyclic shifts of the vector $C^{(i)}(\lambda)$. We can omit the superscript $(i)$ if it is not relevant.
\end{description}
\end{definition}

We denote
\begin{equation}\label{eq4:Delta}
  \Delta_{\Rk,\mu}^{(\lambda_1,\lambda_2)}\triangleq\Ps^+_{\Rk,\mu}(\lambda_1)-\Ps^+_{\Rk,\mu}(\lambda_2).
\end{equation}

\begin{theorem}\label{th4:2orbits}
We consider the ring $\Z_\Rk$ and a representation of its elements $\lambda$ by a sum of $\mu$ elements of $\Z_\Rk$.
Let $\widehat{\Ps}^+_{\Rk,\mu}(\lambda)$, $\Ps^+_{\Rk,\mu}(\lambda)$ be as in  Definition \ref{def3:Lambda} with \eqref{eq3:hatPk lambda}, \eqref{eq3:Pk lambda}. Let $G_{\Rk,\mu}$ be as in \eqref{eq4:G_mu}.
Let $\Dk_{\Rk,\mu},\mu,\Rk,\Phi_\mu,\Phi_{\Rk}$ be as in \eqref{eq4:DRkmu}, \eqref{eq4:DRkmu2}. Let the ch$(\lambda)$-vector $C^{(i)}(\lambda)$, its cyclic shift $C_\delta^{(i)}(\lambda)$, and the number $m^{(i)}(\lambda)$ of distinct cyclic shifts of $C^{(i)}(\lambda)$ be as in Definition \ref{def4:ch-vector} with \eqref{eq4:ch-vectorView}--\eqref{eq4:ch-vectorCyclic}. Let $\Dk_{\Rk,\mu}\ge2$ be a prime. Assume that in $\Z_\Rk$ under $G_{\Rk,\mu}$, there are exactly two orbits $\Oc_0=\Os^\oplus_0$ and $\Oc_1=\Z_\Rk\setminus\Oc_0$, see Theorem \ref{th4:orbits}(i)(vi)(vii) with \eqref{eq4:orbitsRkj}--\eqref{eq4:D=muPrimeO1}.
 Then the following holds:
\begin{description}
  \item[(i)]
  \begin{align}\label{eq4:2orbits}
&\Phi_\Rk\Ps^+_{\Rk,\mu}(0)+(\Rk- \Phi_\Rk)\Ps^+_{\Rk,\mu}(1)=\binom{\Rk}{\mu},~\Ps^+_{\Rk,\mu}(0)-\Ps^+_{\Rk,\mu}(1)=\Delta_{\Rk,\mu}^{(0,1)}.\db\\
&\Ps^+_{\Rk,\mu}(0)=\frac{1}{\mu}\binom{\Rk-1}{\mu-1}-\frac{\Delta_{\Rk,\mu}^{(0,1)}}{\Dk_{\Rk,\mu}}
+\Delta_{\Rk,\mu}^{(0,1)},\label{eq4:2orbits x,y}\db\\
&\Ps^+_{\Rk,\mu}(1)=\frac{1}{\mu}\binom{\Rk-1}{\mu-1}-\frac{\Delta_{\Rk,\mu}^{(0,1)}}{\Dk_{\Rk,\mu}}.\notag\db\\
&\Ps^+_{\Rk,\mu}(\lambda)=\Ps^+_{\Rk,\mu}(0)\text{ if }\lambda\in\Oc_0=\Os^\oplus_0,\,
\lambda\equiv 0\pmod{\Dk_{\Rk,\mu}}~(\text{see }\eqref{eq4:orbitsRkj}, \eqref{eq4:D=muPrimeO0});
\label{eq4:Plambda=P0,Plambda=P1}\db\\
&\Ps^+_{\Rk,\mu}(\lambda)=\Ps^+_{\Rk,\mu}(1)\text{ if }\lambda\in\Oc_1=\Z_\Rk\setminus\Oc_0,\,
\lambda\not\equiv 0\pmod{\Dk_{\Rk,\mu}}~(\text{ see }
\eqref{eq4:D=muPrimeO0}, \eqref{eq4:D=muPrimeO1}).\notag
\end{align}
  \item[(ii)]  We consider a $\mu$-set $\{\lambda_1,\ldots,\lambda_\mu\}\in\widehat{\Ps}^+_{\Rk,\mu}(\lambda)$. The corresponding ch$(\lambda)$-vector is $C^{(i)}(\lambda)=(N_0^{(i)}(\lambda),N_1^{(i)}(\lambda),\ldots,N_{\Dk_{\Rk,\mu}-1}^{(i)}(\lambda))$. We have
\begin{align}\label{eq4:A}
&\sum_{j=1}^\mu\lambda_j\equiv\lambda\pmod{\Dk_{\Rk,\mu}};\db\\
&\sum_{\xi=0}^{\Dk_{\Rk,\mu}-1}\xi N_\xi^{(i)}(\lambda)\equiv\lambda\pmod{\Dk_{\Rk,\mu}}\label{eq4:B}.
\end{align}

  \item[(iii)]
For $\delta=1,2,\ldots,\Dk_{\Rk,\mu}$, the cyclic shift $C_\delta^{(i)}(\lambda)$, of a ch$(\lambda)$-vector $C^{(i)}(\lambda)$ is again a ch$(\lambda)$-vector.

\item[(iv)] The total number $\Ns^\Sigma(C^{(i)}(\lambda))$ of $\mu$-sets $\{\lambda_1,\ldots,\lambda_\mu\}$, belonging to $\widehat{\Ps}^+_{\Rk,\mu}(\lambda)$ and having the same ch$(\lambda)$-vector $C^{(i)}(\lambda)=(N_0^{(i)}(\lambda),N_1^{(i)}(\lambda),\ldots,N_{\Dk_{\Rk,\mu}-1}^{(i)}(\lambda))$ or its cyclic shift $C_\delta^{(i)}(\lambda)$ \eqref{eq4:ch-vectorCyclic}, $\delta=1,2,\ldots,\Dk_{\Rk,\mu}$, is as follows:
  \begin{equation}\label{eq4:N Sigma}
   \Ns^\Sigma(C^{(i)}(\lambda))=\frac{m^{(i)}(\lambda)}{\Phi_\Rk}\prod_{\xi=0}^{\Dk_{\Rk,\mu}-1}\binom{\Phi_\Rk}{N_\xi^{(i)}(\lambda)}.
  \end{equation}
\end{description}
\end{theorem}

\begin{proof}
\begin{description}
  \item[(i)]
 In total, there are $\binom{\Rk}{\mu}$ distinct $\mu$-tuples $\{\lambda_1,\ldots,\lambda_{\mu}\}$ of distinct  elements of  $\Z_\Rk$. By \eqref{eq4:orbitsRkj}, \eqref{eq4:D=muPrimeO0}, the orbit $\Os^\oplus_0$ contains $\Phi_\Rk$ elements of $\Z_\Rk$, hence the orbit $\Oc_1=\Z_\Rk\setminus\Oc_0$ contains $\Rk- \Phi_\Rk$ elements of this ring. By Theorem \ref{th4:Lambda==}(i), for all elements $\lambda$ of an orbit $\Oc_j$ under $G_{\Rk,\mu}$, values of $\Ps_{\Rk,\mu}^+(\lambda)$ are identical. Also we note that, by \eqref{eq4:orbitsRkj}, $0\in\Os^\oplus_0=\Oc_0$ and $1\in\Os^\oplus_1\subseteq\Oc_1$ as $\Dk_{\Rk,\mu}\ge2$. So, we have proved the 1-st equation of the linear system \eqref{eq4:2orbits}. The 2-nd one holds according to \eqref{eq4:Delta}.

 In \eqref{eq4:2orbits x,y}, the values of $\Ps^+_{\Rk,\mu}(0)$ and $\Ps^+_{\Rk,\mu}(1)$ are the solution of the system \eqref{eq4:2orbits}.

 The assertion \eqref{eq4:Plambda=P0,Plambda=P1} follows from Theorem \ref{th4:Lambda==}(i).

  \item[(ii)] By Definition \ref{def3:Lambda}, $\sum_{j=1}^\mu\lambda_j\equiv\lambda\pmod{\Rk}$. By the hypotheses,  $\Dk_{\Rk,\mu}|\Rk$ and $\Dk_{\Rk,\mu}\ge2$ is a prime. This implies the assertion \eqref{eq4:A}. The relation \eqref{eq4:B} follows from \eqref{eq4:A} and Definition \ref{def4:ch-vector} with  \eqref{eq4:ch-vectorView}, \eqref{eq4:ch-vectorDetails}.

  \item[(iii)] It is sufficient to consider the case $\delta=1$.

  Let a $\mu$-set $\{\lambda_1,\ldots,\lambda_\mu\}$ belong to $\widehat{\Ps}^+_{\Rk,\mu}(\lambda)$, see Definition \ref{def3:Lambda} with \eqref{eq3:hatPk lambda}. The corresponding ch$(\lambda)$-vector is $C^{(i)}(\lambda)=(N_0^{(i)}(\lambda),N_1^{(i)}(\lambda),\ldots,N_{\Dk_{\Rk,\mu}-1}^{(i)}(\lambda))$.

      By the hypothesis, $\Dk_{\Rk,\mu}|\mu$. So, $\mu\equiv0\pmod{\Dk_{\Rk,\mu}}$. By \eqref{eq4:A}, $\sum_{j=1}^\mu\lambda_j\equiv\lambda\pmod{\Dk_{\Rk,\mu}}$ that implies $\sum_{j=1}^\mu(\lambda_j+1)\equiv(\lambda+\mu)\pmod{\Dk_{\Rk,\mu}}\equiv\lambda\pmod{\Dk_{\Rk,\mu}}$. Thus, the $\mu$-set $\{\lambda_1+1,\ldots,\lambda_\mu+1\}$ also belongs to $\widehat{\Ps}^+_{\Rk,\mu}(\lambda)$. The corresponding ch$(\lambda)$-vector is $C^{(j)}(\lambda)=(N_0^{(j)}(\lambda),N_1^{(j)}(\lambda),\ldots,N_{\Dk_{\Rk,\mu}-1}^{(j)}(\lambda))$.
      Obviously, $N_0^{(j)}(\lambda)=N_{\Dk_{\Rk,\mu}-1}^{(i)}(\lambda))$, $N_\xi^{(j)}(\lambda)=N_{\xi-1}^{(i)}(\lambda)$, $\xi=1,\ldots.\Dk_{\Rk,\mu}-1$. Thus, $C^{(j)}(\lambda)=C^{(i)}_1(\lambda)$.

\item[(iv)] By the hypotheses,  $\Dk_{\Rk,\mu}|\Rk$ and $\Dk_{\Rk,\mu}\ge2$ is a prime.  Therefore, the ring $\Z_\Rk$ is partitioned into $\Dk_{\Rk,\mu}$ $\Phi_\Rk$-subsets, each of which consists of elements congruent with some fixed $\xi$ modulo $\Dk_{\Rk,\mu}$.
    Moreover, the values of $\xi$ are distinct for distinct subsets. This means that for every $\xi\in\{0,1,\ldots,\Dk_{\Rk,\mu}-1\}$, the ring $\Z_\Rk$ contains exactly $\Phi_\Rk=\Rk/\Dk_{\Rk,\mu}$ elements congruent with $\xi$ modulo  $\Dk_{\Rk,\mu}$. Finally, we add the factor $m^{(i)}(\lambda)$ to take into account the distinct cyclic shifts of $C^{(i)}(\lambda)$.
\end{description}
\end{proof}

\begin{theorem} \label{th4:mu d cases}
Let the hypotheses be as in Theorem \ref{th4:2orbits}. Let $t$ be an integer. We have the following:
\begin{align}\label{eq4:mu3}
&\hspace{0.7cm}\textbf{\emph{(i)}}~\Rk=3t\equiv0\pmod3\ge6,~\mu=3,~\Dk_{\Rk,\mu}=(\Rk,3)=3,~\Phi_\Rk=\frac{\Rk}{3},\db\\
&\Ps^+_{\Rk,3}(0)=\frac{1}{3}\left(\binom{\Rk-1}{2}-1\right)+1,~
\Ps^+_{\Rk,3}(1)=\frac{1}{3}\left(\binom{\Rk-1}{2}-1\right).\notag\db\\
&\hspace{0.7cm}\textbf{\emph{(ii)}}~\Rk=4t+2\equiv2\pmod4\ge6,~\mu=4,\Dk_{\Rk,~\mu}=(\Rk,4)=2,~\Phi_\Rk=\frac{\Rk}{2},
\label{eq4:mu4d2}\db\\
&\Ps^+_{\Rk,4}(0)=\frac{1}{4}\left(\binom{\Rk-1}{3}+\frac{\Rk-2}{2}\right),~\Ps^+_{\Rk,4}(1)=
\frac{1}{4}\left(\binom{\Rk-1}{3}-\frac{\Rk-2}{2}\right).\notag\db\\
&\hspace{0.7cm}\textbf{\emph{(iii)}}~\Rk=5t\equiv0\pmod5\ge10,~\mu=5,~\Dk_{\Rk,\mu}=(\Rk,5)=5,~\Phi_\Rk=\frac{\Rk}{5},\label{eq4:mu5}\db\\
&\Ps^+_{\Rk,5}(0)=\frac{1}{5}\left(\binom{\Rk-1}{4}-1\right)+1,~
\Ps^+_{\Rk,5}(1)=\frac{1}{5}\left(\binom{\Rk-1}{4}-1\right).\notag\db\\
&\hspace{0.7cm}\textbf{\emph{(iv)}}~\Rk=6t+2=2\equiv2\pmod6\ge8\text{ and }\Rk=6t+4\equiv2\pmod6\ge10,
\label{eq4:mu6d2}\db\\
&\mu=6,~\Dk_{\Rk,\mu}=(\Rk,6)=2,~\Phi_\Rk=\frac{\Rk}{2}\equiv1,2\pmod3,\db\notag\\
&\Ps^+_{\Rk,6}(0)=\frac{1}{6}\left(\binom{\Rk-1}{5}-\binom{\Rk/2-1}{2}\right),~
\Ps^+_{\Rk,6}(1)=\frac{1}{6}\left(\binom{\Rk-1}{5}+\binom{\Rk/2-1}{2}\right). \notag\db\\
&\hspace{0.7cm}\textbf{\emph{(v)}}~\Rk=6t+3\equiv3\pmod6\ge9,~\mu=6,~\Dk_{\Rk,\mu}=(\Rk,6)=3,~\Phi_\Rk=\frac{\Rk}{3},
\label{eq4:mu6d3}\db\\
&\Ps^+_{\Rk,6}(0)=\frac{1}{3}\left(\frac{1}{2}\binom{\Rk-1}{5}+\frac{\Rk-3}{3}\right),~
\Ps^+_{\Rk,6}(1)=\frac{1}{6}\left(\binom{\Rk-1}{5}-\frac{\Rk-3}{3}\right). \notag\db\\
&\hspace{0.7cm}\textbf{\emph{(vi)}}~\Rk=7t\equiv0\pmod7\ge14,~\mu=7,~\Dk_{\Rk,\mu}=(\Rk,7)=7,~\Phi_\Rk=\frac{\Rk}{7},\label{eq4:mu7}\db\\
&\Ps^+_{\Rk,7}(0)=\frac{1}{7}\left(\binom{\Rk-1}{6}-1\right)+1,~
\Ps^+_{\Rk,7}(1)=\frac{1}{7}\left(\binom{\Rk-1}{6}-1\right).\notag\db\\
&\hspace{0.7cm}\textbf{\emph{(vii)}}~\Rk=4t+6\equiv2,6\pmod8\ge10,~\mu=8,~\Dk_{\Rk,\mu}=(\Rk,8)=2,
\Phi_\Rk=\frac{\Rk}{2},\label{eq4:mu8d2}\db\\
&\Ps^+_{\Rk,8}(0)=\frac{1}{8}\left(\binom{\Rk-1}{7}+\binom{\Rk/2-1}{3}\right),~
\Ps^+_{\Rk,8}(1)=\frac{1}{8}\left(\binom{\Rk-1}{7}-\binom{\Rk/2-1}{3}\right).\notag\db\\
&\hspace{0.7cm}\textbf{\emph{(viii)}}~\Rk=9t+3\equiv3\pmod9\ge12 \text{ and }\Rk=9t+6\equiv6\pmod9\ge15,\label{eq4:mu9d3}\db\\
&\mu=9,~\Dk_{\Rk,\mu}=(\Rk,9)=3,~\Phi_\Rk=\frac{\Rk}{3}\equiv1,2\pmod3,\notag\db\\
&\Ps^+_{\Rk,8}(0)=\frac{1}{9}\left(\binom{\Rk-1}{8}+2\binom{\Rk/3-1}{2}\right),~
\Ps^+_{\Rk,8}(1)=\frac{1}{9}\left(\binom{\Rk-1}{8}-\binom{\Rk/3-1}{2}\right).\notag
\end{align}
\end{theorem}

\begin{proof}
Table \ref{tab4:1}, for cases (i)--(v), and  Table \ref{tab4:2}, for cases (vii)--(viii), (except the last column) are created using Theorem \ref{th4:2orbits}.
\begin{table}[htbp]
  \caption{The hypotheses are as in Theorem \ref{th4:2orbits}. For $\lambda=0,1$, $\mu=3-6$, integer $t$, and prime $\Dk\triangleq\Dk_{\Rk,\mu}=(\Rk,\mu)\ge2$, all possible ch$(\lambda)$-vectors  $C^{(i)}(\lambda)$ (up to cyclic shifts) and the corresponding values of $\Ns^\Sigma\triangleq\Ns^\Sigma(C^{(i)}(\lambda))$, $\Delta^{(0,1)}_{\Rk,\mu}$. $\Phi\triangleq\Phi_\Rk=\Rk/\Dk_{\Rk,\mu}$, $m\triangleq m^{(i)}(\lambda)$}
  \centering
  \renewcommand\arraystretch{1.15}
 \begin{tabular}
 {c|c|c|c|c|ccccc|c|l|c} \hline
  &&&&&\multicolumn{5}{c|}{ch$(\lambda)$-vector $C^{(i)}(\lambda)$}&&\\
     $\Rk$&$\mu$&$\Dk$&$\lambda$&$i$&$N_0^{(i)}$&$N_1^{(i)}$&$N_2^{(i)}$&$N_3^{(i)}$&$N_4^{(i)}$
     &$m$&$\Ns^\Sigma$&$\Delta^{(0,1)}_{\Rk,\mu}$\\  \hline
     $3t\ge6$&3&3&0&1 &3 & 0& 0 &- &-  &3&$\frac{3}{\Phi}\binom{\Phi}{3}$\\
     &&  &  &4&1 & 1& 1 &- &-  &1&$\Phi^2$&1\\\cline{4-12}

     &&&1&1 &2 & 1& 0 &- &-  &3&$3\binom{\Phi}{2}$\\\hline\hline

     $4t+2$&4&2&  0 &1&4 &0  & -&-&- &2&$\frac{2}{\Phi}\binom{\Phi}{4}$\\
     $\ge6$&&& &3&2 &2  & -&-&-      &1&$\frac{1}{\Phi}\binom{\Phi}{2}^2$&$\frac{\Phi-1}{2}=$\\\cline{4-12}

     &&&  1&  1 &3 &1  & -&-&- &2&$2\binom{\Phi}{3}$&$\frac{\Rk-2}{4}$\\\hline\hline

 $5t$&5&5&0&1&5&0&0&0&0&5&$\frac{5}{\Phi}\binom{\Phi}{5}$\\
 $\ge10$    & & & &6&3&1&0&0&1&5&$5\Phi\binom{\Phi}{3}$\\
     & & &&11&3&0&1&1&0&5&$5\Phi\binom{\Phi}{3}$\\
     & & &&16&2&1&2&0&0&5&$5\binom{\Phi}{2}^2$\\
     & & &&21&2&2&0&1&0&5&$5\binom{\Phi}{2}^2$\\
     & & &&26&1&1&1&1&1&1&$\Phi^4$&1\\\cline{4-12}

     & & &  1&  1 &4 &1  & 0&0&0 &5&$5\binom{\Phi}{4}$\\
     &&& &6   &3 &0  & 1&0&1    &5&$5\Phi\binom{\Phi}{3}$\\
     &&& &11  &3 &0  & 0&2&0    &5&$\frac{5}{\Phi}\binom{\Phi}{3}\binom{\Phi}{2}$\\
     &&& &16  &2 &1  & 1&1&0    &5&$5\Phi^2\binom{\Phi}{2}                      $\\
     &&& &21  &2 &2  & 0&0&1    &5&$5\binom{\Phi}{2}^2                  $\\\hline \hline

$6t+2$&6&2&0&1&6&0&-&-&-&2&$\frac{2}{\Phi}\binom{\Phi}{6}$\\
$\ge8$&&& &  3&4&2&-&-&-&2&$\frac{2}{\Phi}\binom{\Phi}{4}\binom{\Phi}{2}$&$-\frac{1}{3}\binom{\Phi-1}{2}=$\\\cline{4-12}

 $6t+4$     &&&1&1&5&1&-&-&-&2&$2\binom{\Phi}{5}$&$-\frac{1}{3}\binom{\Rk/2-1}{2}$\\
 $\ge10$     &&& &  3&3&3&-&-&-&1&$\frac{1}{\Phi}\binom{\Phi}{3}^2$  \\\hline\hline

$6t+3$&6&3&0&1&6&0&0&-&-&3&$\frac{3}{\Phi}\binom{\Phi}{6}  $\\
$\ge9$      &&& &  4&4&1&1&-&-&3&$3\Phi\binom{\Phi}{4}           $\\
      &&& &  7&3&3&0&-&-&3&$\frac{3}{\Phi}\binom{\Phi}{3}^2$\\
      &&& & 10&2&2&2&-&-&1&$\frac{1}{\Phi}\binom{\Phi}{2}^3$&$\frac{\Phi-1}{2}=$\\\cline{4-12}

      & & &1&1&5&1&0&-&-&3&$3\binom{\Phi}{5}                            $&$\frac{\Rk-3}{6}$\\
      &&& &  4&4&0&2&-&-&3&$\frac{3}{\Phi}\binom{\Phi}{4}\binom{\Phi}{2}$\\
      &&& &  7&3&2&1&-&-&3&$3\binom{\Phi}{3}\binom{\Phi}{2}             $\\\hline
       \end{tabular}
  \label{tab4:1}
\end{table}

\begin{table}[htbp]
  \caption{The hypotheses are as in Theorem \ref{th4:2orbits}. For $\lambda=0,1$, $\mu=8,9$, integer $t$, and prime $\Dk\triangleq\Dk_{\Rk,\mu}=(\Rk,\mu)\ge2$, all possible ch$(\lambda)$-vectors  $C^{(i)}(\lambda)$ (up to cyclic shifts) and the corresponding values of $\Ns^\Sigma=\Ns^\Sigma(C^{(i)}(\lambda))$, $\Delta^{(0,1)}_{\Rk,\mu}$. $\Phi\triangleq\Phi_\Rk=\Rk/\Dk_{\Rk,\mu}$, $m\triangleq m^{(i)}(\lambda)$}
  \centering
  \renewcommand\arraystretch{1.2}
 \begin{tabular}
 {c|c|c|c|c|ccc|c|l|c} \hline
  &&&&&\multicolumn{3}{c|}{ch$(\lambda)$-vector $C^{(i)}(\lambda)$}&&\\
     $\Rk$&$\mu$&$\Dk$&$\lambda$&$i$&$N_0^{(i)}$&$N_1^{(i)}$&$N_2^{(i)}$
     &$m$&$\Ns^\Sigma$&$\Delta^{(0,1)}_{\Rk,\mu}$\\  \hline

$4t+6$&8&2&0&1&8&0&-&2&$\frac{2}{\Phi}\binom{\Phi}{8}$\\
$\ge10$      &&  & &3&6&2&-&2&$\frac{2}{\Phi}\binom{\Phi}{6}\binom{\Phi}{2}$&$\frac{1}{4}\binom{\Phi-1}{3}=$\\
      &&  & &5&4&4&-&1&$\frac{1}{\Phi}\binom{\Phi}{4}^2$&$\frac{1}{4}\binom{\Rk/2-1}{3}=$\\\cline{4-10}

      & & &1&1&7&1&-&2&$2\binom{\Phi}{7}$&$\frac{(\Rk-2)(\Rk-4)(\Rk-6)}{192}$\\
      & & & &3&5&3&-&2&$\frac{2}{\Phi}\binom{\Phi}{5}\binom{\Phi}{3}$&$=\sum_{u=1}^tu^2$\\\hline\hline

$9t+3$&9&3&0&1&9&0&0&3&$\frac{3}{\Phi}\binom{\Phi}{9}$\\
$\ge12$& & & &4&7&1&1&3&$3\Phi\binom{\Phi}{7}$\\
$9t+6$      & & & &7&6&3&0&3&$\frac{3}{\Phi}\binom{\Phi}{6}\binom{\Phi}{3}$\\
$\ge15$      & & &&10&6&0&3&3&$\frac{3}{\Phi}\binom{\Phi}{6}\binom{\Phi}{3}$\\
      & & &&13&5&2&2&3&$\frac{3}{\Phi}\binom{\Phi}{5}\binom{\Phi}{2}^2$\\
      & & &&16&4&4&1&3&$3\binom{\Phi}{4}^2$\\
      & & &&19&3&3&3&1&$\frac{1}{\Phi}\binom{\Phi}{3}^3$&$\frac{1}{3}\binom{\Phi-1}{2}=$\\\cline{4-10}

      & & &1&1&8&1&0&3&$3\binom{\Phi}{8}$&$\frac{1}{3}\binom{\Rk/3-1}{2}$\\
      & & & &4&7&0&2&3&$\frac{3}{\Phi}\binom{\Phi}{7}\binom{\Phi}{2}$\\
      & & & &7&6&2&1&3&$3\binom{\Phi}{6}\binom{\Phi}{2}$\\
      & & &&10&5&1&3&3&$3\binom{\Phi}{5}\binom{\Phi}{3}$\\
      & & &&13&5&4&0&3&$\frac{3}{\Phi}\binom{\Phi}{5}\binom{\Phi}{4}$\\
      & & &&16&4&3&2&3&$\frac{3}{\Phi}\binom{\Phi}{4}\binom{\Phi}{3}\binom{\Phi}{2}$\\\hline
  \end{tabular}
  \label{tab4:2}
\end{table}

The values of $\Ps^+_{\Rk,\mu}(0)$, $\Ps^+_{\Rk,\mu}(1)$ can be obtained directly from the next-to-last column $\Ns^\Sigma$ of the  tables as the sum of all the corresponding elements of this column. Then the values of $\Delta^{(0,1)}_{\Rk,\mu}$ can be obtained from $\Ps^+_{\Rk,\mu}(0)$ and $\Ps^+_{\Rk,\mu}(1)$ after simple (in general, cumbersome) transformations; for example,
\begin{align*}
&\textbf{(i)}~\Ps^+_{\Rk,3}(0)=\frac{3}{\Phi_\Rk}\binom{\Phi_\Rk}{3}+\Phi_\Rk^2=
\frac{3}{\Phi_\Rk}\cdot\frac{\Phi_\Rk(\Phi_\Rk-1)(\Phi_\Rk-2)}{2\cdot 3}+\Phi_\Rk^2\db\\
&=\frac{\Phi_\Rk^2-3\Phi_\Rk+2+2\Phi_\Rk^2}{2}=3\binom{\Phi_\Rk}{2}+1,~\Ps^+_{\Rk,3}(1)=3\binom{\Phi_\Rk}{2},
~\Delta^{(0,1)}_{\Rk,3}=1.\db\\
&\textbf{(ii)}~\Ps^+_{\Rk,4}(0)=\frac{2}{\Phi_\Rk}\binom{\Phi_\Rk}{4}+\frac{1}{\Phi_\Rk}\binom{\Phi_\Rk}{2}^2=
\frac{2}{\Phi_\Rk}\cdot\frac{\Phi_\Rk(\Phi_\Rk-1)(\Phi_\Rk-2)(\Phi_\Rk-3)}{2\cdot3\cdot4},\db\\
&+\frac{1}{\Phi_\Rk}\cdot\frac{\Phi_\Rk^2(\Phi_\Rk-1)^2}{2^2}=\frac{(\Phi_\Rk-1)((\Phi_\Rk-2)(\Phi_\Rk-3)
+3\Phi_\Rk(\Phi_\Rk-1))}{12}\db\\
&=\frac{(\Phi_\Rk-1)(4\Phi_\Rk^2-8\Phi_\Rk+6)}{12}=\frac{(\Phi_\Rk-1)(2\Phi_\Rk^2-4\Phi_\Rk+3)}{6},~
\db\\
&\Ps^+_{\Rk,4}(1)=2\binom{\Phi_\Rk}{3}=\frac{(\Phi_\Rk-1)((2\Phi_\Rk^2-4\Phi_\Rk)}{6},~\Delta^{(0,1)}_{\Rk,4}=\frac{\Phi_\Rk-1}{2}=\frac{\Rk-2}{4}.
\end{align*}
For cases (iii)--(v), (vii)--(viii), the values of $\Delta^{(0,1)}_{\Rk,\mu}$, written in the last column of the tables, can be obtained similarly to cases (i),(ii). Also, for case (vi), we have $\Delta^{(0,1)}_{\Rk,7}=1$; this result we obtained by the same way, but in the tables we do not give the details for case (vi) to save space.

 Finally, we substitute the values of $\Delta^{(0,1)}_{\Rk,\mu}$  into \eqref{eq4:2orbits x,y} and obtain relatively simple relations for $\Ps^+_{\Rk,\mu}(0)$, $\Ps^+_{\Rk,\mu}(1)$, written in the statement of this theorem.
\end{proof}

\begin{conjecture}\label{conj4}
 Let $\mu\ge11$ be a prime. Let $\Dk_{\Rk,\mu}=(\Rk,\mu)$. Then
 \begin{equation}\label{eq4:conj}
\Delta^{(0,1)}_{\Rk,\mu}=1,~ \Ps^+_{\Rk,\mu}(0)=\frac{1}{\mu}\left(\binom{\Rk-1}{\mu-1}-1\right)+1,~
\Ps^+_{\Rk,\mu}(1)=\frac{1}{\mu}\left(\binom{\Rk-1}{\mu-1}-1\right).
 \end{equation}
\end{conjecture}
For $\mu=3,5,7$, see Theorem \ref{th4:mu d cases}.

The next theorem illustrates that when the number of orbits under $G_{\Rk,\mu}$ is greater than two, we can use the same approach as for two orbits.

\begin{theorem}\label{th4:mu4D4}
We consider the ring $\Z_\Rk$ and a representation of its elements $\lambda$ by a sum of $\mu$ elements of $\Z_\Rk$.
Let $\widehat{\Ps}^+_{\Rk,\mu}(\lambda)$, $\Ps^+_{\Rk,\mu}(\lambda)$ be as in  Definition \ref{def3:Lambda} with
\eqref{eq3:hatPk lambda}, \eqref{eq3:Pk lambda}. Let $G_{\Rk,\mu}$ be as in \eqref{eq4:G_mu}. Let
$\Dk_{\Rk,\mu},\Phi_{\Rk},\Delta_{\Rk,\mu}^{(\lambda_1,\lambda_2)}$ be
as in \eqref{eq4:DRkmu}, \eqref{eq4:DRkmu2}, \eqref{eq4:Delta}; let
\begin{equation}\label{eq4:mu4D4}
 \Rk=4t\equiv0\pmod4\ge8\text{ with integer }t,~\mu=4,~\Dk_{\Rk,4}=4,~\Phi_{\Rk}=\Rk/4=t.
\end{equation}
Let the ch$(\lambda)$-vector $C^{(i)}(\lambda)$, its cyclic shift $C_\delta^{(i)}(\lambda)$, and the number $m^{(i)}(\lambda)$ of distinct cyclic shifts of $C^{(i)}(\lambda)$ be as in Definition \ref{def4:ch-vector} with \eqref{eq4:ch-vectorView}--\eqref{eq4:ch-vectorCyclic}. Then the following holds:

Under $G_{\Rk,4}$, there are exactly three orbits $\Oc_0=\Os^\oplus_0$, $\Oc_1=\Os^\oplus_1\cup\Os^\oplus_3$, and $\Oc_2=\Os^\oplus_2$, see Theorem \ref{th4:orbits}(i) with \eqref{eq4:orbitsRkj}.
\begin{align}\label{eq4:mu4D4lambda}
&\Ps^+_{\Rk,4}(\lambda)=\Ps^+_{\Rk,4}(0)\text{ if }\lambda\in\Oc_0=\Os^\oplus_0,~\lambda\equiv0\pmod4;\db\\
&\Ps^+_{\Rk,4}(\lambda)=\Ps^+_{\Rk,4}(1)\text{ if }\lambda\in\Oc_1=\Os^\oplus_1\cup\Os^\oplus_3,
~\lambda\equiv1,3\pmod4,~\lambda\equiv1\pmod2;\notag\db\\
&\Ps^+_{\Rk,4}(\lambda)=\Ps^+_{\Rk,4}(2)\text{ if }\lambda\in\Oc_2=\Os^\oplus_2,~\lambda\equiv2\pmod4.\notag
\end{align}

The values of $\Ps^+_{\Rk,4}(\lambda)$, $\lambda=0,1,2$, satisfy the system
\begin{align}\label{eq4:mu4D4 system}
&\left\{\begin{array}{c}
\Ps^+_{\Rk,4}(0)+2\Ps^+_{\Rk,4}(1)+\Ps^+_{\Rk,4}(2)=\binom{R-1}{3}\\
\Ps^+_{\Rk,4}(0)-\Ps^+_{\Rk,4}(1)=\Delta^{(0,1)}_{\Rk,4}=(\Rk-4)/4\\
\Ps^+_{\Rk,4}(0)-\Ps^+_{\Rk,4}(2)=\Delta^{(0,2)}_{\Rk,4}=-1
                     \end{array}
\right.,
\end{align}
the solution of which is as follows:
\begin{align}\label{eq4:solution}
&\Ps^+_{\Rk,4}(0)=\frac{1}{4}\left(\binom{\Rk-1}{3}+\frac{\Rk-6}{2}\right),\db\\
&\Ps^+_{\Rk,4}(1)=\Ps^+_{\Rk,4}(0)-\frac{\Rk-4}{4}=\frac{1}{4}\left(\binom{\Rk-1}{3}-\frac{\Rk-2}{2}\right),\notag\db\\
&\Ps^+_{\Rk,4}(2)=\Ps^+_{\Rk,4}(0)+1=\frac{1}{4}\left(\binom{\Rk-1}{3}+\frac{\Rk+2}{2}\right).\notag
\end{align}
\end{theorem}

\begin{proof}
The assertion on three orbits follows from Theorem \ref{th4:orbits}(iii). Also, Table \ref{tab4:3} confirms this.
Table \ref{tab4:3} (except the last column) is created using Theorem \ref{th4:2orbits} similarly to
Tables \ref{tab4:1}, \ref{tab4:2}. Then, in the last column, the values of $\Delta^{(0,1)}_{\Rk,4}$, $\Delta^{(0,2)}_{\Rk,4}$ are obtained similarly to Proof of Theorem \ref{th4:mu d cases}.
\begin{table}[htbp]
  \caption{For $\Rk=4t\equiv0\pmod4\ge8$, $\mu=4$, $\Dk=\Dk_{\Rk,4}=(\Rk,4)=4$, and $\lambda=0,1,2,3$, all possible ch$(\lambda)$-vectors  $C^{(i)}(\lambda)$ (up to cyclic shifts) and the corresponding values of $\Ns^\Sigma=\Ns^\Sigma(C^{(i)}(\lambda))$, $\Delta^{(0,1)}_{\Rk,4}$,  $\Delta^{(0,2)}_{\Rk,4}$; $t$ is an integer, $\Phi=\Phi_\Rk=\Rk/4=t$, $m=m^{(i)}(\lambda)$}
    \centering
  \renewcommand\arraystretch{1.2}
 \begin{tabular}
 {c|c|c|c|c|cccc|c|l|c} \hline
  &&&&&\multicolumn{4}{c|}{ch$(\lambda)$-vector $C^{(i)}(\lambda)$}&&&$\Delta^{(0,1)}_{\Rk,4}$\\
  $\Rk$   &$\mu$&$\Dk$&$\lambda$&$i$&$N_0^{(i)}$&$N_1^{(i)}$&$N_2^{(i)}$&$N_3^{(i)}$
     &$m$&$\Ns^\Sigma$&$\Delta^{(0,2)}_{\Rk,4}$\\  \hline

     $4t$&4&4&0&1&4&0&0&0&4&$\frac{4}{\Phi}\binom{\Phi}{4}     $\\
   $\ge8$   &&&  &5&2&0&2&0&2&$\frac{2}{\Phi}\binom{\Phi}{2}^2$&$\Delta^{(0,1)}_{\Rk,4}=$\\
      &&&  &7&2&1&0&1&4&$4\Phi\binom{\Phi}{2}              $&$\Phi-1=$\\\cline{4-11}

      && &1&1&3&1&0&0&4&$4\binom{\Phi}{3}    $&$\frac{\Rk-4}{4}$\\
      &&&  &5&2&0&1&1&4&$4\Phi\binom{\Phi}{2}$\\\cline{4-12}

      && &2&1&3&0&1&0&4&$4\binom{\Phi}{3}               $\\
      &&&  &5&2&2&0&0&4&$\frac{4}{\Phi}\binom{\Phi}{2}^2$&$\Delta^{(0,2)}_{\Rk,4}=$\\
      &&&  &9&1&1&1&1&1&$\Phi^3                         $&$-1$\\\cline{4-12}

      && &3&1&3&0&0&1&4&$4\binom{\Phi}{3}$\\
      &&&  &5&2&1&1&0&4&$4\Phi\binom{\Phi}{2}$\\\hline
  \end{tabular}
  \label{tab4:3}
\end{table}

The assertion \eqref{eq4:mu4D4lambda} follows from Theorem \ref{th4:Lambda==}(i), see also
\eqref{eq4:orbitsRkj}--\eqref{eq4:D=muPrimeO1}, \eqref{eq4:Plambda=P0,Plambda=P1}.

In total there are $\binom{\Rk}{4}$ distinct $4$-tuples $\{\lambda_1,\lambda_2,\lambda_3,\lambda_4\}$ of distinct  elements of  $\Z_\Rk$. By \eqref{eq4:orbitsRkj}, $\#\Os^\oplus_j=\Phi_\Rk$. Now, together with the assertion on three orbits and Theorem \ref{th4:Lambda==}(i) we have
$\Phi_\Rk\Ps^+_{\Rk,4}(0)+2\Phi_\Rk\Ps^+_{\Rk,4}(1)+\Phi_\Rk\Ps^+_{\Rk,4}(2)=\binom{R}{4}$ that implies the first equation of system \eqref{eq4:mu4D4 system}. The 2-nd and 3-rd equations use Table \ref{tab4:3}. The solution \eqref{eq4:solution} is obvious.
\end{proof}

\section{The weight distribution of cosets of weight 2 of the $[q+1,q+2-d,d]_q$ GDRS codes of distance $d\ge5$}\label{sec5:WD2}
We define $\lambda(\gamma)$ as follows:
  \begin{equation}\label{eq5:lambdagamma}
\gamma\in\F_q,~\lambda(\gamma)\in\Z_{q-1},~
  \beta_q^{\lambda(\gamma)}=\gamma,~\beta_q\text{ is a primitive element of }\F_q.
  \end{equation}

\begin{theorem}\label{th5:WD w2coset}
\textbf{(the weight distribution of a coset of weight $2$ of a GDRS code)}
  Let $\Pb^\times_{q,\mu}(\gamma)$ and $\Ps^+_{\Rk,\mu}(\lambda)$ be as in  Definitions \ref{def3:Gamma} and \ref{def3:Lambda}. Let $\C$ be a $[q+1,q+2-d,d]_q$ normalized GDRS code of distance $d\ge5$. Let $\V^{(2)}$ be one of its coset of weight $2$ with a leader $\vf_2(j_1,j_2;\gamma_1,\gamma_2)$. Let $\lambda(\gamma)$ be as in \eqref{eq5:lambdagamma}.
Let $A_w(\C)$  be as in \eqref{eq2:Aw(C)}. Let, as in  \eqref{eq2:OmegaDef},
   \begin{align}\label{eq5:Omega}
&\Omega_w^{(0)}(q+1,d)=(-1)^{w-d}\binom{q+1}{w}\binom{w-1}{d-2}=(-1)^{w-d}\binom{q+1}{q+1-w}\binom{w-1}{w-d+1},\db\\
&\Omega_w^{(2)}(q+1,d)=(-1)^{w-d}\binom{q-1}{w-2}\binom{w-3}{d-4}=(-1)^{w-d}\binom{q-1}{q+1-w}\binom{w-3}{w-d+1}.\notag
   \end{align}
   Then the weight distribution of the coset $\V^{(2)}$ does not depend on the positions  $j_1,j_2$; the number $B_w(\V^{(2)})$ of vectors of weight $w$ in $\V^{(2)}$  has the  form:
\begin{align}\label{eq5:W=2}
&\textbf{\emph{(i)}}~~ B_w(\V^{(2)})=0 \text{ if } w\in\{0,1,\ldots,d-3\}\setminus\{2\};~B_2(\V^{(2)})=1;\db\\
&B_{d-2}(\V^{(2)})=\Ps^+_{q-1,d-2}(\lambda(-\gamma_2/\gamma_1))
=\Pb^\times_{q,d-2}(-\gamma_2/\gamma_1);\db\notag\\
&B_w(\V^{(2)})=A_w(\C)-\Omega_w^{(0)}(q+1,d)+\Omega_w^{(2)}(q+1,d)\db\notag\\
&+(-1)^{w-d}\binom{q+3-d}{q+1-w}\cdot\Ps^+_{q-1,d-2}(\lambda(-\gamma_2/\gamma_1))\text{ if }w= d-1,d,\ldots,q+1.\notag
\end{align}

\textbf{\emph{(ii)}} If for all $\lambda$ from $\Z_{q-1}$ the value of $\Ps^+_{q-1,d-2}(\lambda)$ is the same or, equivalently, all  $\binom{n}{2}(q-1)^2$ cosets $\V^{(2)}$ of weight $2$ have the same weight distribution,
then
\begin{align}\label{eq5:w=2 ident}
&B_w(\V^{(2)})=0 \text{ if } w\in\{0,1,\ldots,d-3\}\setminus\{2\};~B_2(\V^{(2)})=1;~\db\\
&B_{d-2}(\V^{(2)})=\frac{1}{q-1}\binom{q-1}{d-2}=\frac{1}{d-2}\binom{q-2}{d-3};~B_w(\V^{(2)})=A_w(\C)-\Omega_w^{(0)}(q+1,d)
\db\notag\\
&+\Omega_w^{(2)}(q+1,d)+
(-1)^{w-d}\binom{q+3-d}{q+1-w}\cdot\frac{1}{q-1}\binom{q-1}{d-2}\text{ if }w= d-1,d,\ldots,q+1.\notag
\end{align}

\textbf{\emph{(iii)}} If $(q-1,d-2)=1$, i.e. $q-1$ and $d-2$ are coprime, then for all $\lambda$ from $\Z_{q-1}$ the value of $\Ps^+_{q-1,d-2}(\lambda)$ is the same, all $\binom{n}{2}(q-1)^2$ cosets $\V^{(2)}$ of weight $2$ have the same weight distribution, and all the relations from \eqref{eq5:w=2 ident} hold. In particular, this means that the code $\C$ is \emph{2-regular}, see Definition \ref{def2:tregul}. Moreover, in this case the values $\frac{1}{q-1}\binom{q-1}{d-2}=\frac{1}{d-2}\binom{q-2}{d-3}$ are integer.
\end{theorem}
\begin{proof}
\textbf{(i)}, \textbf{(ii)}  The assertions follow from \eqref{eq2:OmegaDef}, \eqref{eq2:Bon_w=2}, \eqref{eq2:wd2_ident}, \eqref{eq3:Bd-2=Gamma}, \eqref{eq3:Gamma=Lamb}, \eqref{eq3:Bd-2=Lamb}, after simple transformations.

\textbf{(iii)} The assertions follow from Theorems \ref{th2:wd1cosetMDS}, \ref{th2:W=2}(iv), \ref{th4:orbits}(iv), \ref{th4:Lambda==}, and case \textbf{(ii)} of this theorem.
\end{proof}

\begin{theorem}\label{th5:cases}
Let $q$ be a prime power. Let $\C$ be the $[q+1,q+2-d,d]_qR$ normalized GDRS code with minimum distance $d\ge5$. Let $\V^{(2)}$ be a coset of weight $2$ of the code $\C$ with a leader $\vf_2(j_1,j_2;\gamma_1,\gamma_2)$. Then for specific pair $q,d$ from the first two columns of Table \ref{tab5:1}, the values of $\Ps^+_{q-1,d-2}(\lambda(-\gamma_2/\gamma_1))$, crucial for the weight distribution of $\V^{(2)}$ according to \eqref{eq5:W=2}, are given in the rest of columns of the table.

\begin{table}[htbp]
  \caption{The values of $\Ps^+_{q-1,d-2}(\lambda(-\frac{\gamma_2}{\gamma_1}))$ for specific pair $q,d$.
   $t$ is an integer, $\gamma=-\frac{\gamma_2}{\gamma_1}$}
    \centering
  \renewcommand\arraystretch{1.2}
 \begin{tabular}
 {c|c|c|l|c|c} \hline
&&&&the condition for &\\
&&&&$\Ps^+_{q-1,d-2}(\lambda(\gamma))=$&refe-\\
$q$&$d$&$j$&$\Ps^+_{q-1,d-2}(j)$&$\Ps^+_{q-1,d-2}(j)$&rence\\\hline

$3t+1\ge7$&5&$0$&$\frac{1}{3}\left(\binom{q-2}{2}-1\right)+1$&$\lambda(\gamma)\equiv0\pmod3$&\eqref{eq4:mu3}\\
&&$1$&$\frac{1}{3}\left(\binom{q-2}{2}-1\right)$&$\lambda(\gamma)\not\equiv0\pmod3$ \\\hline

$4t+3\ge7$&6&$0$&$\frac{1}{4}\left(\binom{q-2}{3}+\frac{q-3}{2}\right)$&$\lambda(\gamma)\equiv0\pmod2$&\eqref{eq4:mu4d2}\\
$$&&$1$&$\frac{1}{4}\left(\binom{q-2}{3}-\frac{q-3}{2}\right)$&$\lambda(\gamma)\not\equiv0\pmod2$ \\\hline

$5t+1\ge11$&7&$0$&$\frac{1}{5}\left(\binom{q-2}{4}-1\right)+1$&$\lambda(\gamma)\equiv0\pmod5$&\eqref{eq4:mu5}\\
&&$1$&$\frac{1}{5}\left(\binom{q-2}{4}-1\right)$&$\lambda(\gamma)\not\equiv0\pmod5$ \\\hline

$6t+3\ge9$&8&$0$&$\frac{1}{6}\left(\binom{q-2}{5}-\binom{(q-3)/2}{2}\right)$&$\lambda(\gamma)\equiv0\pmod2$&\eqref{eq4:mu6d2}\\
$6t+5\ge11$&&$1$&$\frac{1}{6}\left(\binom{q-2}{5}+\binom{(q-3)/2}{2}\right)$&$\lambda(\gamma)\not\equiv0\pmod2$ \\\hline

$6t+4\ge10$&8&$0$&$\frac{1}{3}\left(\frac{1}{2}\binom{q-2}{5}+\frac{q-4}{3}\right)$&$\lambda(\gamma)\equiv0\pmod3$
&\eqref{eq4:mu6d3}\\
           &&$1$&$\frac{1}{6}\left(\binom{q-2}{5}-\frac{q-4}{3}\right)$&$\lambda(\gamma)\not\equiv0\pmod3$ \\\hline

$7t+1\ge29$&9&$0$&$\frac{1}{7}\left(\binom{q-2}{6}-1\right)+1$&$\lambda(\gamma)\equiv0\pmod7$&\eqref{eq4:mu7}\\
&&$1$&$\frac{1}{7}\left(\binom{q-2}{6}-1\right)$&$\lambda(\gamma)\not\equiv0\pmod7$ \\\hline

$4t+7\ge11$&10&$0$&$\frac{1}{8}\left(\binom{q-2}{7}+\binom{(q-3)/2}{3}\right)$&$\lambda(\gamma)\equiv0\pmod2$&\eqref{eq4:mu8d2}\\
&&$1$&$\frac{1}{8}\left(\binom{q-2}{7}-\binom{(q-3)/2}{3}\right)$&$\lambda(\gamma)\not\equiv0\pmod2$ \\\hline

$9t+4\ge13$&11&$0$&$\frac{1}{9}\left(\binom{q-2}{8}+2\binom{(q-4)/2}{2}\right)$
&$\lambda(\gamma)\equiv0\pmod3$&\eqref{eq4:mu9d3}\\
$9t+7\ge16$&&$1$&$\frac{1}{9}\left(\binom{q-2}{8}-\binom{(q-4)/3}{2}\right)$&$\lambda(\gamma)\not\equiv0\pmod3$ \\\hline

$4t+1\ge9$&6&$0$&$\frac{1}{4}\left(\binom{q-2}{3}+\frac{q-7}{2}\right)$&$\lambda(\gamma)\equiv0\pmod4$
&\eqref{eq4:solution}\\
&&$1$&$\frac{1}{4}\left(\binom{q-2}{3}-\frac{q-3}{2}\right)$&$\lambda(\gamma)\equiv1\pmod2$& \eqref{eq4:mu4D4lambda}\\
&&$2$&$\frac{1}{4}\left(\binom{q-2}{3}+\frac{q+1}{2}\right)$&$\lambda(\gamma)\equiv2\pmod4$& \\\hline
\end{tabular}
\label{tab5:1}
\end{table}
\end{theorem}

\begin{proof}
  The assertion follows from Corollary \ref{cor3:Bd-2} with \eqref{eq3:Bd-2=Lamb}, Theorems \ref{th4:orbits} with \eqref{eq4:orbitsRkj}--\eqref{eq4:D=muPrimeO1}, \ref{th4:2orbits} with \eqref{eq4:2orbits}--\eqref{eq4:N Sigma},
\ref{th4:mu d cases} with \eqref{eq4:mu3}--\eqref{eq4:mu6d3}, \ref{th4:mu4D4} with \eqref{eq4:mu4D4}--\eqref{eq4:solution}, and  \ref{th5:WD w2coset} with \eqref{eq5:Omega}--\eqref{eq5:w=2 ident}. In particular, for the results of Section \ref{sec4:orbits} we have $\Rk=q-1,~\mu=d-2,\,\Dk_{\Rk,\mu}=(q-1,d-2)$.
\end{proof}

\begin{conjecture}\label{conj5}
 Let $q=pt+1$ be a prime power, where $p\ge3$ is a prime, $t\ge2$ is an integer. Let $\C$ be the $[q+1,q+2-d,d]_qR$ normalized GDRS code with minimum distance $d=p+2$. Let $\V^{(2)}$ be a coset of weight $2$ of the code $\C$ with a leader $\vf_2(j_1,j_2;\gamma_1,\gamma_2)$.  Then the values of $\Ps^+_{q-1,d-2}(\lambda(-\gamma_2/\gamma_1))=\Ps^+_{pt,p}(\lambda(-\gamma_2/\gamma_1))$, crucial for the weight distribution of $\V^{(2)}$ according to \eqref{eq5:W=2}, are as follows:
 \begin{align}\label{eq5:conj}
&\Ps^+_{q-1,d-2}\left(\lambda\left(-\frac{\gamma_2}{\gamma_1}\right)\right)=\frac{1}{p}\left(\binom{q-2}{p-1}-1\right)+1
\text{ if }\lambda\left(-\frac{\gamma_2}{\gamma_1}\right)\equiv0\pmod p;\db\\
&\Ps^+_{q-1,d-2}\left(\lambda\left(-\frac{\gamma_2}{\gamma_1}\right)\right)=\frac{1}{p}\left(\binom{q-2}{p-1}-1\right)
\text{ if }\lambda\left(-\frac{\gamma_2}{\gamma_1}\right)\not\equiv0\pmod p.\notag
 \end{align}
\end{conjecture}
Note that for $p=3,5,7$, Conjecture \ref{conj5} is proved in Theorem \ref{th5:cases} with Table \ref{tab5:1}.

\section*{Acknowledgments}
The research of A.A. Davydov was carried out within the state assignment of Ministry of Science and Higher Education of the Russian Federation  No.\ FFNU-2025-0028 for Kharkevich Institute for Information Transmission Problems. A.A. Davydov would like to thank E. Jagudaeva for organizational support.
The research of S. Marcugini and F. Pambianco was supported in part by the Italian
National Group for Algebraic and Geometric Structures and their Applications (GNSAGA -
INDAM) (Contract No. U-UFMBAZ-2019-000160, 11.02.2019) and by University of Perugia
(Project No. 98751: Strutture Geometriche, Combinatoria e loro Applicazioni, Base Research
Fund 2017--2019; Fighting Cybercrime with OSINT, Research Fund 2021). This work was partially funded by the SERICS project (PE00000014) under the MUR National Recovery and Resilience Plan funded by the European Union-NextGenerationEU.

\end{document}